\RequirePackage{fix-cm}
\documentclass[smallextended]{svjour3}       
\smartqed  
\usepackage{graphicx}
\usepackage{latexsym}
\usepackage{array,geometry}
\usepackage{amssymb,amstext,amsfonts,amsmath}
\usepackage{enumerate}
\usepackage{epsfig}

\newcommand{\R}{\mathbb R}
\newcommand{\Z}{\mathbb Z}
\newcommand{\aff}{\mathcal{A}}
\newcommand{\affb}{\mathcal{B}}
\newcommand{\affm}{\mathcal{M}}
\newcommand{\GL}{\rm{GL}}
\DeclareMathOperator{\dist}{dist}
\DeclareMathOperator*{\argmin}{arg\,min}
\begin{document}

\title{Study of a model for reference-free plasticity
}


\author{Stephan Luckhaus       \and
   Jens Wohlgemuth      
}


\institute{S. Luckhaus  \at
               Universit\"at Leipzig\\
              Tel: 0341 97321 08\\
              Fax: 0341 97321 95\\
              \email{Stephan.Luckhaus@math.uni-leipzig.de}           
           \and
           J. Wohlgemuth \at
           Max- Planck institute for mathematics in science\\
           Tel: 	0341 9959 969\\
           \email{jwohlgem@mis.mpg.de}
}

\date{Received: date / Accepted: date}

\maketitle

\begin{abstract}

 We investigate a Kac-type many particle model that allows a reference-free description of plastic deformation.
 In the framework of the model a solid body is described by a set of particle positions. A lattice is fitted to the particle configuration around each point on a  mesoscopic scale. The lattice parameters are used as an argument of a non-linear elasticity energy functional. 
There are two main results in this paper. First, we prove an estimate for the difference between the fitted lattice parameters of points of low energy density that are  sufficiently close to each other. Sequences of these points can be used for homotopy type arguments. In particular it is possible to identify dislocations as topological defects in this framework. 
Furthermore, we use the fitted lattice parameters as local Lagrangian coordinates and bound the energy from below with a functional of these coordinates.

\end{abstract}

\section{Introduction:}
In this article we discuss a many particle Hamiltonian that allows a description of plastic deformations without using a reference configuration.
The model is closely related to the one presented by L. Mugnai and S. Luckhaus in \cite{Lucapaper}.
In the classical theory of elasticity the deformation of a solid body is described with the help of a reference configuration, that is assumed be stress-free.
The actual configuration of the described body is given as the image of this reference configuration by a differentiable map $\phi$.
The energy of a configuration is then given by 
\begin{equation}
H=\int_\Omega \tilde{F}\left(\nabla\phi(z)\right)dz
\end{equation}  
In this setting the deformed configuration is the minimizer of this energy functional under certain boundary conditions.
However, the local order is fixed by the reference configuration. Since plastic deformations are changing the local order, they can not be described in this framework.
Therefore we are aiming to substitute the reference configuration in the framework with a quantity that allows a change of the local order. If we imagine the reference configuration filled with the lattice $\Z^d$.
These position are mapped on $\phi(\Z^d)$. In the neighborhood of a points $z$ it holds 
\begin{equation}
\phi(z_i)\approx \phi(z)+\nabla \phi(z) (z_i-z)
\end{equation}
Hence, the configuration is approximately a Bravais lattice in the neighborhood of each point.
The main idea of our model is to fit a Bravais lattice locally to a set of atom positions and use the matrix that spans the Bravais lattice as an argument for an elastic energy functional.
In this paper will demonstrate that chains of theses fitted lattices can be used to define a generalized Burgers vector that characterizes the topological defects of a crystal. 
Furthermore, we prove that the fitted lattice parameters can be used as Lagrangian coordinates. And that we can bound the energy density from below with a functional of this coordinates. In the form  $h_\lambda \ge \tilde{F}(\nabla \tau)+ C \|\nabla^2\tau\|^2$.

\begin{figure}[ht]
\center
\includegraphics[scale=0.5]{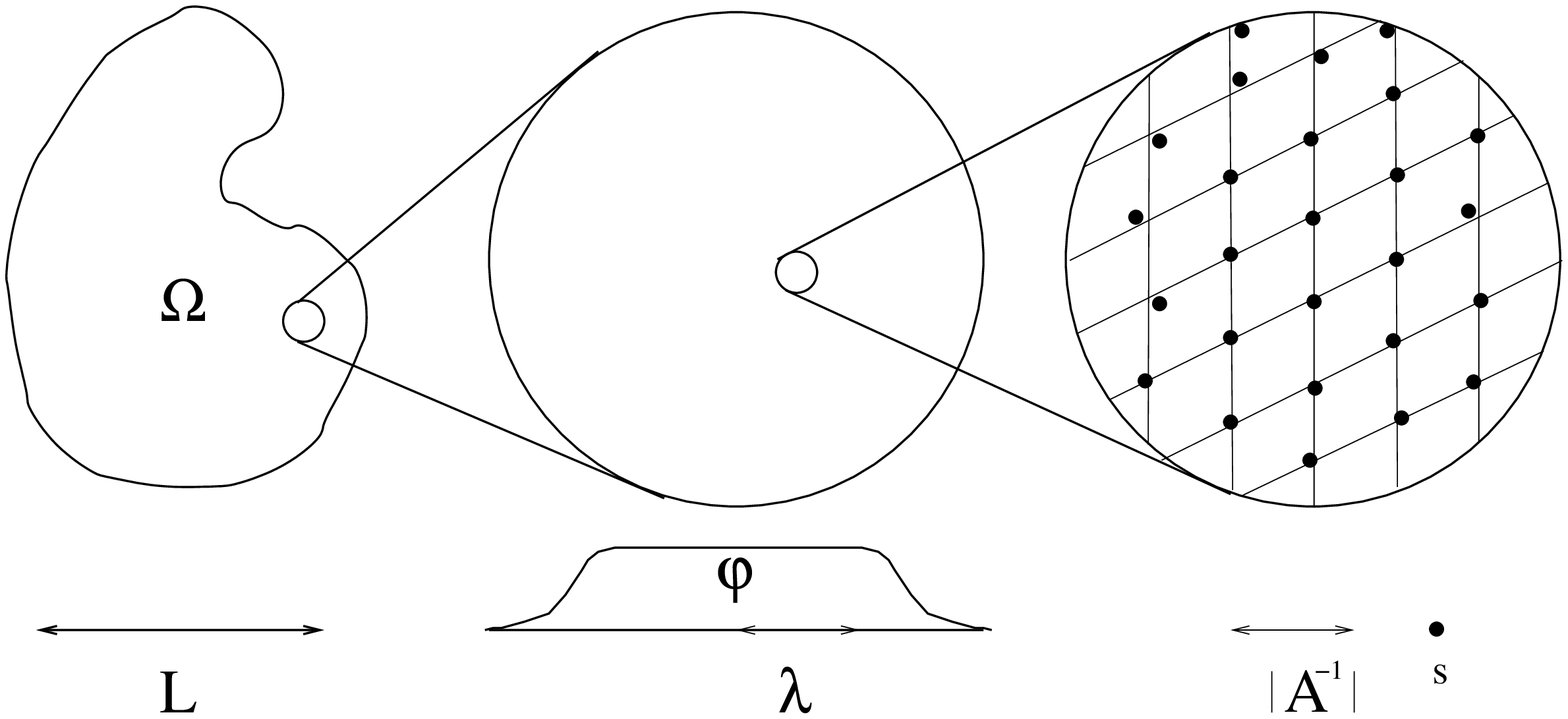}
\caption{Multi-scale model with three different scales: Microscopic scale: $|A^{-1}|$ distance between atoms, macroscopic scale $L$ size of the body , mesoscopic scale: $\lambda$ the configuration looks like a lattice}
\end{figure}

Eventually we hope that a connection to non-equilibrium statistical mechanics can be established. If at low temperatures a strong enough bound for our Hamiltonian  can be derived for a statistical mechanic particle system, then there is the hope to use the local fields described in this paper as the thermodynamic quantities of the system.

\section{Definition of the model:}

In our model the actual state of the described body is given by a domain $\Omega\subset \R^d$ and a set of atom positions $\chi=\{x_i\in B_{4\lambda}(\Omega)| i=1...N\}$ , where $\lambda$ is the mesoscopic scale $\lambda<<L$. Here $d$ denotes the dimension. 
The set of atoms $\chi$ consists of two subsets $\chi=\chi_I\cup\chi_S$.
The internal atoms $\chi_I\subset \Omega$ can move freely inside $\Omega$, but are not allowed to leave it.
The boundary atoms $\chi_S\subset B_{4\lambda}\left(\Omega\right)/\Omega$ are fixed and serve as our boundary condition.
We call the number of internal atoms $N_I=\sharp \chi_I$ and the number of boundary atoms $N_S=\sharp \chi_S$.
The energy in our model is given by an integral over an energy density and an hardcore particle interaction $V$ with radius $s_0$.
\begin{equation}\label{ModelGleichung1}
H_{\lambda}\left( \chi \right):=\int_{B_{2\lambda}(\Omega)}\hat{h}_{\lambda}\left( \chi ,x\right )+\sum_{i,j} V\left(|x_i-x_j|\right).
\end{equation}
The main part of the model is the energy density $\hat{h}_{\lambda}\left( \chi,x \right )$ in Eulerian coordinates $x$.
This density is determined by fitting a Bravais lattice. 
$\chi_{\left(A,\tau \right)}+x=A^{-1}(\Z^d-\tau)+x$ locally to the atom positions $\chi$, where $A\in Gl_d(\R)$ and $\tau \in \R^d$. We denote: $\aff=(A, \tau)$. For every $\aff$ one can calculate a pre-energy density $h_{\lambda}\left(\aff, \chi, x\right )$ at a given point. 
The energy density $\hat{h}_{\lambda}\left( \chi,x \right )$ is then given by
\begin{equation}\label{Defdichte}
\hat{h}_{\lambda}\left( \chi, x \right ):=\inf_{\aff}\left\{h_{\lambda}\left(\aff, \chi, x\right )  \right\}\quad .
\end{equation}
The pre-energy density $h_{\lambda}\left(\aff, \chi. x\right )$ consists of three parts.
\begin{equation}
h_{\lambda}\left(\aff, \chi, x\right ):=F\left(A\right)+J_{\lambda}\left(\aff, \chi, x\right) +\nu_{\lambda}\left(A, \chi,x\right)
\end{equation}
The first term $F$ measures the elastic contribution to the energy and corresponds to the energy density in the classical theory.
The second part $J_{\lambda}$ measures energy cost of deviations of the configuration $\chi$ from the fitted lattice .
The last part $\nu_{\lambda}$ assigns a cost to the vacancies.
In the following we will explain the properties of the different parts of the energy density in more detail.

$F\left(A\right)$ is related to $\tilde{F}$ of the classical theory with the formula $F(G)=\tilde F(G^{-1})\det(G^{-1})$ for the transformation between Eulerian and Lagrangian coordinates.
We want to consider $F\in C_2\left(Gl_d(\R)\right)$ with the following properties 
\begin{enumerate}[1)]
\item $F(A)=F(A R)$,  $\forall A\in Gl_d(\R)$, $\forall R\in SO_d$ \hfill  (Frame indifference)
\item $\exists E\in Gl_d(\R)$ with $F(E)=0$ \hfill (Existence of minimizer)
\item $F(A)\ge C_1^{El}\left(\det(E)-\det(A))\right)^2+C_2^{El}\text{dist}^2\left(A, E\; SO_d \right)$ \hfill(Coercivity)
\end{enumerate} 
for some $C_1^{El},C_2^{El}>0$. We use the Euclidean norm to define the distance for two matrices $\text{dist}(A,E)=|A-E|$. 
$J_{\lambda}\left(\aff, \chi,x \right)$ uses the affine transformation $\aff(x)=Ax+\tau$ to map the atom positions in the $\lambda$-neighborhood of the position $x$ into a periodic potential $W$ with minima in $\Z^d$. $W$ is assumed to be locally convex around the minima.
In this way $J_\lambda$ is approximately the standard deviation of the configuration $\chi$ from the fitted lattice $\chi_\aff+x$.
\begin{equation}
J_{\lambda}\left(\aff, \chi,x\right):=\frac{\left\|A^{-1}\right\|^2}{C_{\varphi}\lambda^d}\sum_{i} W(A\left(x_i-x\right)+\tau)\varphi\left(\lambda^{-1}\left|x_i-x\right|\right)
\end{equation}
where $\varphi\in C^\infty\left(\R^+\right)$ is a smooth and monotone decreasing cut-off function and has the following properties
\begin{enumerate}[1)]
\item $\varphi(x)=1$ for $x\le1$
\item $\varphi(x)=0$ for $x\ge2$
\item $\partial_x \varphi \leq 0$ 
\end{enumerate}
We use $C_{\varphi}:=\int_{\R^d} \varphi(\left|x\right|) dx$ as a normalization constant.
We also use the notation $\tilde{\varphi}(x):=\varphi(|x|)$.
We assume that the periodic potential $W\in C^\infty\left(\R^d\right)$ fulfills 
\begin{enumerate}[1)]
\item $W(z)=W(z+z_n) $    $\forall z_n \in \Z^d, \forall z \in \R^d $\, \hfill  (Periodicity)
\item $ c_\Theta^0 y^2 \leq y\nabla^2W(x)y\leq  c_\Theta^1 y^2 $    $\forall y \in \R^d, x\in B_{\Theta_W}(\Z^d)$\hfill (Local convexity)
\item $C^W_0 \mathrm{dist}^2(z,\mathbb{Z}^d)\le W(z)\le C^W_1 \mathrm{dist}^2(z,\mathbb{Z}^d)$ \hfill(Coercivity)
\item $W(z)=W(-z)$    $\forall z \in \R^d $ \quad\quad\quad\quad(Symmetry)
\end{enumerate}
where  $\Theta_W, c_\Theta^0, c_\Theta^1, C^W_0, C^W_1>0$ are constants. We define the local density of a configuration $\chi$ by
\begin{equation}
\rho_\lambda(\chi,x):=\frac{1}{C_{\varphi}\lambda^d}\sum_i \varphi\left(\lambda^{-1}\left|x_i-x\right|\right)
\end{equation}
Moreover, we define :
\begin{equation}
\nu_{\lambda}\left(A,\chi,x\right):=\vartheta \left|\det A-\rho_\lambda(\chi,x)\right|
\end{equation}
Therefore, the energy per vacancies is $\vartheta$. This part also ensures that a lattice that is finer than necessary will not be fitted to the configuration because it would contain a big number of vacancies.
$V: \R^+\rightarrow \left\{0, \infty\right\}$ is an hard core repulsion. It has the technical purpose, to prevent several atoms from sitting at the same lattice side.
\begin{equation}
 V\left(x\right):=
 \begin{cases}
 0& \text{for}\qquad x\ge s_0\\
 \infty & \text{for}\qquad x< s_0 .
 \end{cases}
\end{equation}
The hard-core potential implies, that any configuration with finite energy has a particle density smaller than $\rho^{max}_d+O(\lambda^{-1})$.
\begin{equation}
\rho^{max}_d=\frac{2^d}{w_ds_0^d}+O(s_o\lambda^{-1})\quad ,
\end{equation}
where $w_d$ is the volume of the $d$-dimensional unit sphere.

\section{Notations and important definitions}
We introduce the following sets:
\begin{align}
Gl_d(\R):=&\left\{A\in \R^{d\times d} | \det A>0     \right\} \quad,\quad
Gl_d(\Z):=\left\{A\in \Z^{d\times d} | \det A=1     \right\},  \nonumber\\
B_{r}(U):=&\left\{x\in \R^{d} | \exists y\in u \; \text{with} |y-x|<r    \right\} . 
\end{align}
For $\aff=(A,\tau)\in \R^{d\times d}\times \R^d$ we use the following norms:  
\begin{align}
\|A\|^2:=&\sum_{i,j=1}^d A_{i,j}^2\quad, \quad 
|A|:=\sup \left\{|A e| | e\in R^d ,|e|=1    \right\}\quad, \nonumber\\
\|\aff\|_\lambda^2:=&\|A\|^2+|\mu|^2\quad .
\end{align}

\begin{definition}\label{Reparametrisation} 
We call a pair $\affb=(B,z)\in Gl_d(\Z)\times  \Z^d$ a reparametrisation.
For $\aff=(A,\tau)\in Gl_d(\Z)\times  \Z^d$ we define the reparametrisation of $\aff$ as
\begin{equation}
\affb \aff= (B A,B\tau+t)
\end{equation}
\end{definition}
We note that $\chi_\aff=\chi_{\affb \aff}$. Hence, Bravais-lattices are invariant under reparametrisations.
Since, we fit Bravais lattices to the atom configuration, the minimizing $\aff$ may jump parametrisations of the same lattice.
These different parametrisations are connected by reparametrisations.

\begin{definition}\label{definitiontopologicalchain}
For a sequence of reparametrisations  $\affb_{j,j+1}=(B_{j,j+1},t_{j,j+1})\in(Gl_d(\Z),\Z^d)$,
we define the product reparametrisation  $\affb_{0,1}=(\textbf{B},\textbf{t})\in Gl_d(\Z)\times \Z^d$ as composition of the affine maps given by the reparametrisations
\begin{align}
\textbf{B}=&B_{0,1}....B_{K-1,K}= \prod_{j=1}^KB_{j-1,j} \quad ,\quad 
\textbf{t}=\sum_{k=1}^K \left(\prod_{j=1}^{k-1} B_{j-1,j}\right)t_{k-1,k} \quad.
\end{align}
\end{definition}

\begin{definition}\label{definitionregularatoms}
For an atom configuration $\chi$ and lattice parameters $\aff=(A,\tau)\in Gl_d(\R)\times\R^d$, and a position $x$ and a distance $\beta>0$, we define
the $(\aff,\beta, x)$-regular atoms and irregular atoms
\begin{align}
\chi_{\aff,\beta, x}^{reg}:=\{x_i\in \chi | dist(x_i, \chi_{\aff}+x) \le \beta  \}\quad,\quad
\chi_{\aff,\beta, x}^{irr}:=\{x_i\in \chi | dist(x_i, \chi_{\aff}+x) > \beta  \},
\end{align}
and the densities of regular atoms and irregular atoms 
\begin{align}
\rho_{\aff,\beta}^{reg}(x) := & \rho_\lambda (\chi_{\aff,\beta ,x}^{reg},x)=\frac{1}{C_{\varphi}\lambda^d} \sum_{x_i\in \chi_{\aff,\beta, x}^{reg}} \varphi\left(\lambda^{-1}\left|x_i-x\right|\right) \quad ,\nonumber\\
\rho_{\aff,\beta}^{irr}(x):= & \rho_\lambda (\chi_{\aff,\beta, x}^{irr},x)=\frac{1}{C_{\varphi}\lambda^d} \sum_{x_i\in \chi_{\aff,\beta, x}^{irr}} \varphi\left(\lambda^{-1}\left|x_i-x\right|\right) \quad .
\end{align}
\end{definition}

Next, we introduce the notion of regular pairs.
\begin{definition}\label{definitionregularpoints}
Let $\aff=(A,\tau)\in Gl_d(\R)\times \R^d$  and $\epsilon_\rho, \epsilon_J,  C_A \in \R$
and let $\chi$ be the configuration, then we say that the pair $(x, \aff)$ is  $(\epsilon_\rho, \epsilon_J, C_A)$-regular, if the following conditions are fulfilled
\begin{enumerate}
\item $\|A^{-1}\|<C_A$\quad ,
\item $\left|\rho_\lambda(\chi,x)-\det A\right|<\epsilon_{\rho} \det A$\quad ,
\item $J_\lambda(\aff,\chi,x)<\epsilon_J \rho_\lambda(\chi,x)$\quad , 
\item $|x_i-x_j|>s_o$ for all $i,j$ \quad .
\end{enumerate}
\end{definition}
If the pair $(x,\aff)$ is regular this means that the configuration looks like the lattice $\chi_\aff+x$ in the $B_{2\lambda}(x)$.
We say a point $x$ is regular, if there exits an $\aff\in \in Gl_d(\R)\times \R^d$ such that the pair $(x,\aff)$ is regular.
Theorem \ref{addingatoms} explains the connection between regular pairs, reparametrisations and the product reparametrisation. 
For regular pairs with $\epsilon_\rho=1/8$ we get 
\begin{equation}
|A|\le |A^{-1}|^{d-1}\det A\le |A^{-1}|^{d-1} \frac{1}{1+\epsilon_{\rho}}\rho \le \frac{C_A^{d-1}}{1-\epsilon_{\rho}}\rho^{max}_d
\le C_{|A|}:= \frac{8C_A^{d}}{7}\rho^{max}_d\quad.
\end{equation}

\section{Main theorems}
 Theorem \ref{addingatoms} says, that in case of 
a sequence of regular pairs $(y_{j}, \aff_j)$ fulfilling $|y_{j-1}-y_{j}|\le \frac{3}{2}\lambda$  the 
affine maps are connected by reparametrisations.																																																									
 The product of these reparametrisations  does not change, if one adds or leaves out  points in the middle of a chain chain. Hence, the reparametrisation product is a topological invariant, determined only by the homotopy class of the chain.
\begin{theorem}\label{addingatoms}
For all $C_A>s_o$ there exists
$ \hat{\lambda}\in \R$ and $\epsilon_J>0$ such that for all $\lambda>\hat{\lambda}$,
 $ \aff_j=(A_j,\tau_j)\in Gl_d(\R)\times \R^d$ and $y_j \in B_{2\lambda}(\Omega)$ with $j=1,...,N$ 
 the following holds:
 \begin{enumerate}
 \item If $(y_j,\aff_j)$ is $(\epsilon_\rho, \epsilon_J, C_A)$-regular for $j=0,...,N$ and
$|y_{j-1}-y_{j}|\le \frac{3}{2}\lambda$ for $j=1,...,N$,
then there exists uniquely defined reparametrisations $\affb_{j-1,j}=\left( B_{j-1,j},t_{j-1,j} \right)\in Gl_d(\Z)\times  \Z^d$ such that 
\begin{align}\label{equationjumping}
\|id-A_{j-1}^{-1}B_{j-1,j}A_{j}\|<& \frac{c^{A}_J}{\sqrt{\det A_{j}}}  \left(\frac{2\lambda}{2\lambda-|y_{j-1}-y_{j}|}\right)^{d/2}\frac{\sqrt{J_{j-1,j}}}{\lambda},  \nonumber\\
\left|\delta \tau_{j-1,j}\right|<&   \frac{c^{\tau}_J|A_{j-1}|}{\sqrt{\det A_{j}}} \left(\frac{2\lambda}{2\lambda-\left|y_{j-1}-y_{j}\right|}\right)^{d/2} \sqrt{J_{j-1,j}},
\end{align}
where
\begin{align}
\delta \tau_{j-1,j}=&B_{j-1,j}\tau_{j}+t_{j-1,j} -\tau_{j-1}-\frac{B_{j-1,j}A_{j}+A_{j-1}}{2}\left(y_{j}-y_{j-1}\right)\nonumber\\
J_{j-1,j}:=&\max\left\{ J_\lambda(\aff_{j-1},\chi,y_{j-1}), J_\lambda(\aff_{j},\chi,y_{j})\right\}\quad ,\nonumber\\
c^{A}_J:=& \frac{3}{2} \left( \frac{8 d C_\varphi} {C_{\varphi 2}C_0^W}\right)^{\frac{1}{2}} \quad ,\quad \quad c^{\tau}_J:=\left(\frac{10}{C_0^W}\right)^{\frac{1}{2}} \quad .
\end{align}
\item If additionally it holds $|y_{k+1}-y_{k-1}|\le \frac{3}{2}\lambda$ for some $k$ then there exists a $\tilde{\affb}_{k-1,k+1}$ fulfilling the estimates \eqref{equationjumping} for the point $j+1$ instead of the point $j$ and we have
\begin{equation}\label{eqtopologicalchain}
\prod_{j=1}^N\affb_{j-1,j}=\left(\prod_{j=1}^{k-1} \affb_{j-1,j} \right) \tilde{\affb} \left(\prod_{j=k+2}^{N} \affb_{j-1,j} \right)
\end{equation}
\end{enumerate}
\end{theorem}

We call regular $(x,\aff_1)$ and $(x,\aff_2)$ equivalent when the $\affb \in Gl(\Z)\times \Z^d$ connecting them as described by theorem \ref{addingatoms} is
$\affb=(Id, 0)$. Hence, we get for every regular $x$ an set off $P^\aff_x$  equivalence classes $[\aff ]$
The group $G=Gl_d(\Z)\times \Z^d$ acts on this set of equivalence classes by the action
$G\times P\aff_x\rightarrow P\aff_x: \affb[\aff ]=[\affb\aff]$.
Furthermore, we know by theorem \ref{addingatoms} that adding or omitting a point in a sequence of regular pairs.
does not change the reparametrisation product. We call two chains equivalent, if they can be deformed into each other by this process.
We denote with $Hom_x$ the set of equivalence classes of this chains with starting point and endpoint $x$. We use these like homotopy classes.
Each $S\in Hom_x$ induces an one to one map $\hat{B}_S: P^\aff_x\leftarrow P^\aff_x$  $B([\tilde{\affb} \aff_0])=[\tilde{\affb} \affb_0 \aff_0  ]$.
where $\aff_0$  is an arbitrary selected so that  $(x,\aff_0)$ is regular and $\affb_0$ is the reparametrisation product of a chain of the equivalence class starting and ending with $(\aff_0)$. We call the map $\hat{B}_S$ the generalized Burgers vector.
We note that if $\affb_0$ would commute with $\tilde{\affb}$, the generalized Burgers vector would be just given be a simple multiplication with $\affb_0$.
Furthermore, we note that the map from $\hat{B} Hom_x \rightarrow  Iso(P^\aff_x,P^\aff_x)$ is an homomorphism of groups.

Compared to the description of the generalized Burgers vector in \cite{Lucapaper} our chains allows us to extend the homotopy classes over thin barriers of irregular points.

Related descriptions of solid bodies can be found in Kondo \cite{Kondo64} and Kr\"oner \cite{Kroener58}(see also  \cite{Davini86},\cite{Davini91}, \cite{Cermelli99} \cite{Ariza05} )

Theorem \ref{TheoremLowerboundoftheensitywiththelocalminimizersofJ} says that the local minimizers $\tilde{\aff}_B $ of $J_\lambda(\cdot,\chi,x)$ are differentiable functions of $x$ and that we can use them as Lagrangian coordinates. Moreover, we can bound the energy density from below with an functional of the form $\tilde{F}(\nabla \tilde{\tau}_B)+ C \|\nabla^2\tilde{\tau}_B\|^2$
\begin{theorem}\label{TheoremLowerboundoftheensitywiththelocalminimizersofJ}
There exists $\hat{\lambda}$, $\hat{\epsilon}>0$ such that for $\lambda>\hat{\lambda}$ 
for all points $x$ with $\hat{h}_\lambda(\chi, y)\le\hat{\epsilon}$ and for all reparametrisations $\affb=(B,t)\in Gl_d(\Z^d)\times \Z^d$ fulfilling $\|\hat{A}^{-1}(x)B^{-1}| \le 2 \|E^{-1}\|$ where $\hat{\aff}=(\hat{A},\hat{\tau})\in Gl_d(\R)\times\R^d$ is the global minimizer of $h_\lambda(\cdot ,\chi,x)$
there exits a open neighborhood $U$ around $x$ and a two times differentiable function $\tilde{\aff}_B U\rightarrow Gl_d(\R)\times\R^d$ with the following properties
\begin{enumerate}
\item 
\begin{equation}
\left\|\tilde\aff_B(x)-\affb\hat{\aff}(x) \right\|_\lambda\le \left(  \frac{1}{2} C_{Con} \|A_0^{-1}\|^2\rho_\lambda \right)^{-1/2}\sqrt{J_\lambda}(\aff_0,\chi,x)
\end{equation}
\item $\tilde\aff_B(y)$ is a local minimizer of $J_\lambda(\cdot, \chi, y)$ for every $y$ in $U$
\item 
\begin{align}\label{Lowerbound}
\hat{h}_\lambda(\chi,y)\ge&F_C\left(\nabla \tilde{\tau}_B(y)\right)+ \frac{1}{5}\tilde{C}_{\nabla }\left(\frac{\rho_{2\lambda}}{\rho_\lambda}\right)     \left\|\nabla\tilde{\tau}_B^{-1}(y)\right\|^2\lambda^4\|\nabla^2\tilde{\tau}_B(y)\|^2\det \left(\nabla \tilde{\tau}_B \right) \quad ,
\end{align}
\end{enumerate}
where we denote
\begin{align}
F_C(A):=&\inf\left\{U(A, A_1, B, A_2)| A_1,A_2\in Gl_d(\R),B\in Gl_d(\Z)\right\}\quad ,\nonumber\\
U(A, A_1, B, A_2):=&F(A_2)+\frac{1}{3} C_{Con}C_{rep}^{-1} \|\left(BA_2\right)^{-1}\|^2 \det \left( A \right)
\lambda^2\left\|BA_2-A_1\right\|^2\nonumber\\
&+\frac{1}{2}\tilde{C}_{\nabla }\left(\frac{\rho_{2\lambda}}{\rho_\lambda}\right) \left\|A_1^{-1}\right\|^2\det \left( A\right)
\lambda^2\|A- A_1 \|^2 \quad ,
\end{align}
where we use the following constants
\begin{align}
 \alpha_\nabla:=&64\max\left\{ \frac{ \|\nabla W\|_{\infty}^2}{C^W_0 \Theta_W^2  }, \frac{|c_\Theta^1|^2}{c_\Theta^0}\right\} 
C_{con}:= c_\Theta^0 \min\left\{\frac{1}{12}  , \frac{c_\Theta^0 C_{\varphi}^2}{4\left(9+d\right)w_{d-1}^2 4^d}\frac{\rho_\lambda^2}{\det A^2} \right\}\quad ,\nonumber\\
C_{rep}=& 9 \left(C^W_0\right)^{-1} 4^{d-1} C_A^{2d} \det E^2   \quad ,\nonumber\\
\tilde{C}_{\nabla }^{-1}\left(X\right):=&C_{rep}\left(C_{\nabla 2}(X)^{-1}+ \frac{\alpha_\nabla 2^d\|\nabla \sqrt{\tilde{\varphi}}\|_\infty^2}{C_{con}^2 X }\right) \quad ,\nonumber\\
\left(C_{\nabla 2}(X)\right)^{-\frac{1}{2}}:=& \frac{\sqrt{\alpha_\nabla}}{C_{con} }\left(\|\nabla\sqrt{\tilde{\varphi}}\|_\infty^2+ \|\nabla^2\sqrt{\tilde{\varphi}}\|_\infty+ 2 \infty \|\nabla \sqrt[4]{\tilde{\varphi}}\|^2\right)d\sqrt{2^d X}\nonumber\\
& +\frac{\sqrt{\alpha_\nabla}} {C_{con}^2 }\left(2^d\|\nabla\sqrt{\tilde{\varphi}}\|_\infty^2\right)^{\frac{1}{2}}\left(16\;  2^{\frac{d}{2}} X +\sqrt{8d}\sqrt{X}\right)
\end{align}
\end{theorem}

\begin{itemize}
\item The function $\tilde{\aff}_B(y(s))$ can be extended along the curve of regular atoms as long as \newline $|\tilde{\aff}_B(y(s))| \le C_A$.
If we start at $\|\hat{A}^{-1}B^{-1}\| \le 3/2 \|E^{-1}\|^{-1}$, we can extend it as least for a distance scaling like $\lambda^2$ areas of low energy density.
\item If we select $\hat{\epsilon}$ small enough, the local minimizer $\tilde{A}$ can not leave the Ericson Piterie neighborhood it started in without increasing the energy over this barrier.
Therefore, in this case $\tilde{\aff}_B$ can be extended in any connected set of low energy points.
\item Due to the coercivity conditions on $F$ it holds:
\begin{equation}
F_C\left(A\right)= \min\left\{ F(BA)|B\in Gl_d(\Z^d)  \right\}+O(\lambda^{-2})\quad .
\end{equation}
\end{itemize}

\section{Ideas of the proofs}
\paragraph{Proof of Theorem \ref{addingatoms}}
If there are two $(\epsilon_\rho, \epsilon_J, C_A)$-regular pairs $(x, \aff_1)$ and $(x,\aff_2)$ for the same point $x$, then $\aff_2$ is a reparametrisation of $\aff_1$up to a small difference (Lemma \ref{Lemmatwodifferentfittedlattices}).  Furthermore, the difference in $A$ can be controlled by $\lambda^{-1}\sqrt{\epsilon_J}$ and the difference in $\tau$ can be controlled by $\sqrt{\epsilon_J}$.
Additionally, we prove in Lemma \ref{Sectpunktbewegenlambdaaendern} that all points in a $\lambda$-ball around a regular point are regular with modified coefficients and a smaller $\tilde{\lambda}$. 
If we combine these lemmata, we get similar estimates for two $(\epsilon_\rho, \epsilon_J, C_A)$-regular pairs $(y_1,\aff_1)$  and $(y_2,\aff_2)$ , provided that $|y_1-y_2|\le 1.5\lambda $. For sufficiently small $\epsilon_J$ the reparametrisation between them will be unique.
Additionally, if we have three regular pairs $(y_i\aff_i)$ with $|y_j-y_k|\le 1.5\lambda $, the reparametrisations fulfills
\begin{align}\label{triangel}
B_{1,3}=&B_{1,2}B_{2,3}\quad ,      \quad t_{1,3}= B_{1,2}t_{2,3}+t_{1,2} \quad .
\end{align}
Therefore, for a sequence of sufficiently regular points satisfying $\left|y_{j+1}-y_j\right|<1.5\lambda$ we get a reparametrisation for every step.
Furthermore, we can conclude from equation \eqref{triangel} that, if we add an additional regular point somewhere in the middle of the sequence, the product of the reparametrisations stays the same. 

\paragraph{Proof of Theorem \ref{TheoremLowerboundoftheensitywiththelocalminimizersofJ} }
This proof is based on the local convexity of $J_\lambda(\cdot,\chi,x)$ for regular $x$, that is proved in Lemma \ref{hjlocalconvexity}.
Using the local convexity we prove in Lemma \ref{localaffgradient} that close to every $\aff$ with $(x,\aff)$
 there is a local minimizer $\tilde{\aff}_B$ of $J_\lambda(\cdot,\chi ,x)$.
Furthermore, we show with implicit function theorem, that the local minimizer $\tilde{\aff}_B$ are differentiable functions both of the position $x$ and the configuration $\chi$ in regular areas of the configuration (Lemma \ref{localaffgradient}). In Lemma \ref{Upperboundgradients} we use a more careful application of implicit function theorem to get a lower bound on $J_\lambda(\tilde{\aff}_B(x),\chi,x)$ of the form 
\begin{align}
J_\lambda\left(\tilde{\aff}_B, \chi, x\right)\ge & C \lambda^2\left( \lambda^2\|\nabla \tilde{A}_B\|^2+\|\nabla \tilde{\tau}_B- \tilde{A}_B  \|^2\right) \quad,\nonumber\\
J_\lambda\left(\tilde{\aff}_B, \chi, x\right)\ge & C \lambda^4 \left(\lambda^2\|\nabla^2 \tilde{A}_B\|^2+ \|\nabla^2\tilde{\tau}_B-\nabla \tilde{A}_B\|^2\right)\quad .\nonumber
\end{align}
Additionally, we prove in Lemma \ref{lowenergyregular} that  for all points $x$ with low energy density there exists a global minimizer $\hat{\aff}(x)$ of $h_\lambda(\cdot, \chi,x)$ and $(x\hat{\aff}(x)$ . If $\hat{\aff}(x)$ is regular, its reparametrisation $\affb\hat{\aff}(x)$ is regular too according to Lemma \ref{BasStadev} Furthermore, we can estimate $J_\lambda\left(\affb \hat{\aff}, \chi, x\right) \ge C J_\lambda(\hat{\aff}(x)  , \chi,x)$.
We use Lemma \ref{Theoremjumping} to prove that points are close enough to each other there are reparametrisations that connect the different global minimizers $\hat{\aff}(x)$ for different $x$ with the same differentiable branch of local minimizers $\tilde{\aff}_B$.
Due to the local convexity, we get the estimate 
\begin{equation}
J_\lambda(\affb \hat{\aff},\chi,x)\ge  J_\lambda(\tilde{\aff_{B}},\chi,x)+ \frac{1}{2} C_{Con} \left(\|\left(\affb \hat{A} \right)^{-1}\|^2+O(\lambda^{-1})\right) \rho_\lambda \left\|\affb \hat{\aff}-\tilde{\aff}_B\right\|_\lambda^2 \quad .\nonumber
\end{equation} 
Hence, for low energy points we can estimate $ J_\lambda(\hat{\aff}(x), \chi ,x)$ the gradient of the local minimizers.
If we put these estimates together and minimize over $\hat{\aff}(x)$, $B(x)$ and $\tilde{A}_B$ we get the estimate \eqref{Lowerbound}.

\section{Proof of Theorem  \ref{addingatoms}  } 

\begin{lemma}\label{Lemmatwodifferentfittedlattices}
For all $C_A>s_0$ exists $\hat{\lambda}\in \R$ and $\epsilon_{\rho}$,$\epsilon_J$ such that for all $\lambda>\hat{\lambda}$, 
 $\aff_i=(A_1,\tau_1),\aff_2=(A_2,\tau_2)\in Gl_d(\Z^d)\times \Z^d$
 and  $x\in B_{2\lambda}(\Omega)$, so that $(x,\aff_1)$ and $(x,\aff_2)$ are $(\epsilon_\rho, \epsilon_J, C_A)$-regular, we have
\begin{align}\label{regularresult}
\|id-A_1^{-1}BA_2\|<&\left(\frac{C_0^WC_{\varphi 2}}{8 d C_\varphi}\det A_2\right)^{-\frac{1}{2}} \frac{\sqrt{J_{max}}}{\lambda} \quad ,\nonumber\\
|B\tau_2+t-\tau_1|<& \| A_1\| \left(\frac{C_0^W}{10}\det A_2\right)^{-\frac{1}{2}} \sqrt{J_{max}} \quad,
\end{align}
where 
\begin{equation}
J_{max}=\max\left\{J_\lambda(\aff_1,\chi, x),J_\lambda(\aff_2,\chi,x)\right\}\quad .
\end{equation}
\end{lemma}
\begin{proof}
We will proceed in two steps. The first step is basically taken from the proof of Theorem 5.12 from \cite{Lucapaper}, where the same statement is proved for a related model. In the second step we improve the estimate for the proportionality constant. 

{\bf Step 1:}
Without lose of generality we will restrict ourselves to the case $x=0$.
We have  $\|A_2^{-1}\|,\|A_2^{-1}\|<C_A$. We take some $\gamma>0$ and use Lemma \eqref{Dichteregular} with \newline $\beta=\sqrt{\frac{\epsilon_J^*}{\gamma C_0^W}}<\min \left\{s_o/2,\frac{1-\epsilon_{\rho}}{2 C_A^{d-1}}\left(\rho^{max}_d \right)^{-1}\right\}$. We get the estimate $\rho_{ \aff_j ,\beta ,x}^{irr}\le  \gamma \rho_\lambda$. We denote by $\chi_{reg}$ the set of atoms that are regular for both $\aff_1$ and $\aff_2$.
 We have that at least a density of $\rho_\lambda(\chi_{reg},x )  (1-2\gamma)\rho_\lambda(\chi,x)$ atoms, that are regular for $\aff_1$ and $\aff_2$.
Due to the regularity condition on the density we know that $(1-\epsilon_\rho)\det A_2 \le \rho_\lambda$. Hence, we get 
\begin{equation}
\rho_\lambda(\chi_{reg},0 )\ge(1-2\gamma)(1-\epsilon_\rho)\det A_2\quad .
\end{equation}
Furthermore, if $\beta  \le \frac{1-\epsilon_{\rho}}{2 C_A^{d-1}}\left(\rho^{max}_d \right)^{-1}\le \frac{1}{2}\|A_j\|^{-1}$ a lattice point can not belong to two different atoms. Therefore, there is a bijection between the atoms of $\chi_{reg}$ and the lattice positions $\chi^{reg}_{\aff_2}$ next to them in $\chi_{\aff_2}+x$. Hence, we get:
\begin{equation}\label{phidifference}
\left|\varphi\left(\lambda^{-1}\left|x_{i2}\right|\right)-\varphi\left(\lambda^{-1}\left|x_i\right|\right)\right|\le  \frac{\|\nabla \varphi\|_\infty}{\lambda} \left|x_{i2}\right|\le \frac{\|\nabla \varphi\|_\infty\beta }{\lambda}\quad.
\end{equation}
If we combine this with the estimate on the density of $\chi^{reg}_{\aff_2}$ from lemma \ref{condistapp}, we obtain:
\begin{align}
\rho_\lambda(\chi^{reg}_{\aff_2},0 )\ge&( (1-2\gamma)(1-\epsilon_\rho)\det A_2 - \frac{\|\nabla \varphi\|_\infty \beta }{\lambda}  \det A_2\quad,\\
\rho_\lambda(\chi_{\aff_2}/\chi^{reg}_{\aff_2},0 )\le &\left(2\gamma+\epsilon_\rho+\det A_2 + \frac{\|\nabla \varphi\|_\infty\beta }{\lambda}\right)  \det A_2\quad
\end{align}
We define $Q:= \frac{A_2}{2|A_2|}[-\lambda,\lambda]^d$.
Therefore, it holds $Q \subseteq B_\lambda(0)$ and all $y\in B_{\lambda}(0)$  fulfill $\varphi\left(\lambda^{-1}|y|\right)=1$.
Hence, we get
\begin{align}
\#\left( Q \cap \chi_{\aff_2}/\chi^{reg}_{\aff_2}\right)\le &C_{\varphi}\lambda^d\left(2\gamma+\epsilon_\rho+2\gamma\epsilon_\rho+ \frac{\|\nabla \varphi\|_\infty \beta }{\lambda}\right)  \det A_2\quad,\nonumber\\
 \#\left( Q \cap \chi^{reg}_{\aff_2} \right) \ge& C_{\varphi}\lambda^d\left(2\gamma+\epsilon_\rho+2\gamma\epsilon_\rho+ \frac{\|\nabla \varphi\|_\infty\beta }{\lambda}\right)  \det A_2\quad .
\end{align}
Finally, for all atoms in $Q \cap \chi^{reg}_{\aff_2}$ holds that there is a atom of $\chi^{reg}$ in distance less of $\beta$ from each of them and a point of
$\chi^{reg}_{\aff_2}$ in distance less of $\beta$ from this atom. Due to triangle inequality it holds for all $x_{i2} \in Q \cap \chi^{reg}_{\aff_2}$ 
\begin{align}
2\beta \ge & \dist\left( x_{i2},  x_i \right)+ \dist\left(  x_i , Q \cap \chi^{reg}_{\aff_1}\right)\ge \dist\left( x_{i2}, Q \cap \chi^{reg}_{\aff_1} \right)\quad ,\nonumber\\
2\beta \ge &\dist\left( A_2^{-1}(z_i-\tau_2), A_1^{-1}(\Z^d-\tau_1 )\right) \quad ,\nonumber\\
2\beta \|A_1\|\ge &\dist\left( A_1A_2^{-1}z_i-A_1A_2^{-1}\tau_2+\tau_1, \Z^d\right) \quad .
\end{align}
Therefore, $A_2\left(Q \cap \chi^{reg}_{\aff_2}\right)$ fulfills the conditions of  Theorem 5.12 from \cite{Lucapaper} for sufficiently small 
$\epsilon_J$ and $\epsilon_\rho$ and sufficiently large $\epsilon_J$.
 Hence, there exists $B\in Gl_d(\Z)$ and $t\in \Z^d$ such that
\begin{equation}\label{lucaestimate}
|A_1A_2^{-1}- B| \le  O\left(\frac{\sqrt{\epsilon_J}}{\lambda}\right)\quad,\quad| \tau_1-B \tau_2-t| \le  O\left(\sqrt{\epsilon_J}\right)\quad. 
\end{equation}

{\bf Step B:}
Now, we improve the constant in the estimate
\begin{align}
 2J_{max} \ge& J_\lambda\left(\aff_1,\chi,0\right)+ J_\lambda\left(\aff_1,\chi,0\right) \nonumber\\
\ge&  \frac{C_0^W}{C_{\varphi}\lambda^d} \sum_{x_i\in \chi_{reg}} (\mathrm{dist}^2(x_i, \chi_{\aff_1} )  +\mathrm{dist}^2(x_i, \chi_{\aff_2} ) )\varphi\left(\lambda^{-1}\left|x_i\right|\right)\nonumber\\
 = & \frac{C_0^W}{C_{\varphi}\lambda^d} \sum_{x_i\in \chi_{reg}} ((x_i-x_{i1})^2  +(x_i-x_{i2})^2 ) \varphi\left(\lambda^{-1}\left|x_i\right|\right) \quad .
 \end{align}
Using $a^2+b^2=\frac{1}{2} (a-b)^2+\frac{1}{2} (a+b)^2$ one gets
\begin{align}
  4J_{max}>& \frac{C_0^W}{C_{\varphi}\lambda^d} \sum_{x_i\in \chi_{reg}} (x_{i1}-x_{i2})^2   \varphi\left(\lambda^{-1}\left|x_i\right|\right)\quad .
\end{align}
We count the $x_{i2}$ instead of the $x_i$ due to the bijection between them and change the argument of $\varphi$ from $\left(\lambda^{-1}\left|x_i\right|\right)$ to $\left(\lambda^{-1}\left|x_{i2}\right|\right)$ paying with an error term that we estimate with the inequality \eqref{phidifference}. We get
\begin{align}
4J_{max}\ge& \frac{C_0^W}{C_{\varphi}\lambda^d} \sum_{x_{i2}\in \chi^{reg}_{\aff_2}} \dist(x_{i2},\chi_{\aff_1})^2   \varphi\left(\lambda^{-1}\left|x_{i2}\right|\right)+
O\left(\frac{\sqrt{\epsilon_J}^3}{\lambda}\right)
\end{align}
We use the notation
\begin{align}
X:=&\frac{C_0^W}{C_{\varphi}\lambda^d} \sum_{x_{i2}\in \chi_{\aff_2}} \mathrm{dist}^2(x_{i2}, \chi_{\aff_1}) \varphi\left(\lambda^{-1}\left|x_{i2}\right|\right) \quad , \nonumber\\
Y:=&\sup_{x_{i2}\in \chi_{\aff_2}\cap B_{2\lambda}(0)} \mathrm{dist}^2(x_{i2}, \chi_{\aff_1}) \quad 
\end{align}
and estimate
\begin{align}\label{Abstandteilgitterhilf1}
4J_{max}\ge&\frac{C_0^W}{C_{\varphi}\lambda^d} \sum_{x_{i2}\in \chi^{reg}_{\aff_2}} \mathrm{dist}^2(x_i, \chi_{\aff_1}) \varphi\left(\lambda^{-1}\left|x_{i2}\right|\right)\nonumber\\
\ge& X-\frac{C_0^W}{C_{\varphi}\lambda^d}
 \sum_{x_{i2}\in\chi_{\aff_2}/\chi^{reg}_{\aff_2}} \mathrm{dist}^2(x_i, \chi_{\aff_1}) \varphi\left(\lambda^{-1}\left|x_{i2}\right|\right)\nonumber\\
\ge& X-
\frac{C_0^W}{C_{\varphi}\lambda^d} \sum_{x_{i2}\in\chi_{\aff_2}/\chi^{reg}_{\aff_2}} \varphi\left(\lambda^{-1}\left|x_{i2}\right|\right)  \sup_{x_{i2}\in \chi_{\aff_2}\cap B_{2\lambda}(0)} \mathrm{dist}^2(x_{i2}, \chi_{\aff_1}) \nonumber\\
\ge& X-(\det A_2-\rho_\lambda(\chi^{reg}_{\aff_2},0)) y \quad .
\end{align}
Due to \eqref{lucaestimate} for sufficiently small $\epsilon_J$ it holds for all $z_i\in \Z^d$ with $A_1^{-1}(z_i-\tau_1)\le 2\lambda$
\begin{align}
\dist \left(x_{i2},A_1^{-1}B(z_i-\tau_1+t)\right)&\le \dist\left( A_2^{-1}(z_i-\tau_2), A_1^{-1}B(z_i-\tau_1+t )\right) \nonumber\\
\le &\|A_2^{-1}\|\dist\left( A_2A_1^{-1}z_i-A_2A_1^{-1}\tau_1+\tau_2-t, z_i\right)\le & O(\sqrt{\epsilon_J})\le \frac{1}{2}
\end{align}
Hence, we get
\begin{align}\label{Abstandteilgitterhilf2}
\mathrm{dist}^2(x_{i2}, \chi_{\aff_1})=& (x_{i2}-A_1^{-1}(z_i-\tau_1))^2 \nonumber\\
=&(x_{i2}-A^{-1}_1(BA_2 x_{i2}+B\tau_2+t -\tau_1))^2\nonumber\\
=&\left((1-A^{-1}BA_2)x_{i2}+A^{-1}_1(B\tau_2+t -\tau_1)\right)^2 \quad .
\end{align}
We set 
\begin{align}\label{regulardelta}
\delta_A=& 1-A^{-1}_1BA_2 \quad ,\quad \delta_\tau=A^{-1}_1(B\tau_2+t -\tau_1) \quad ,
\end{align}
and obtain
\begin{align}
Y=&\sup_{x_{i2}\in \chi_{\aff_2}\cap B_{2\lambda}(0)}(\delta_A x_{i2}+\delta_\tau)^2 < (\|\delta_A\|2\lambda+\left|\delta_\tau\right|)^28\lambda^2(\|\delta_A\|)^2+2\left|\delta_\tau\right|^2 \quad .
\end{align}
 Using the equation \eqref{Abstandteilgitterhilf2} we get
 \begin{equation}\label{regular1}
  X = \frac{C_0^W}{C_{\varphi}\lambda^d} \sum_{x_{i2}\in \chi_{\aff_2}}  (\delta_A x_{i2}+\delta_\tau)^2\varphi\left(\lambda^{-1}\left|x_{i2}\right|\right) \quad .
 \end{equation}
Next, we estimate the sum in equation \eqref{regular1} by an integral using Lemma \ref{condistapp} and get
 \begin{equation}
 X>\frac{C_0^W}{C_{\varphi}\lambda^d} \int_{\R^d} (\delta_A\tilde{y}+\delta_\tau)^2\varphi\left(\lambda^{-1}\left|\tilde{y}\right|\right)\det A_2 d\tilde{y}+O\left(\frac{\|\delta_\aff\|_\lambda^2}{\lambda^2}\right)\quad .
 \end{equation}
  We substitute $y=\frac{\tilde{y}}{\lambda}$ and obtain
 \begin{equation}
 X>\frac{C_0^W}{C_{\varphi}} \int_{\R^d}  (\delta_A\lambda y+\delta_\tau)^2\varphi\left(\left|y\right|\right)\det A_2 dy+O\left(\frac{\|\delta_\aff\|_\lambda^2}{\lambda^2}\right) \quad .
\end{equation} 
The integral of an odd function over an even area is zero. Hence, the mixed term vanishes
  \begin{equation}
  X > \frac{C_0^W}{C_{\varphi}\lambda^d}\int  \left((\delta_A y)^2+(\delta_\tau)^2\right)
 \varphi\left(\lambda^{-1}\left|y\right|\right) dz \det A_2 +O\left(\frac{\|\delta_\aff\|^2_\lambda}{\lambda^2}\right) \quad .  
 \end{equation}
 The symmetric matrix $\delta_A^+ \delta_A$ has $d$ eigenvalues $a_1..a_d$
 In the eigensystem of $(\delta_A^+ \delta_A)$ we get
  \begin{align}
 \int  (\delta_A y)^2 \varphi\left(\lambda^{-1}\left|y\right|\right)=&\int \sum_{k=1}^d a_k y_k^2 \varphi\left(\lambda^{-1}\left|y\right|\right) dz
 =\sum_{k=1}^d a_k \int y_k^2 \varphi\left(\lambda^{-1}\left|y\right|\right) dz\nonumber\\
 =& Tr(\delta_A^+\delta_A)\frac{1}{d}  \lambda^d \int \lambda^2 y^2 \varphi\left(\left|y\right|\right) dz
 =\frac{C_{\varphi2}}{d}\lambda^{d+2}\|\delta_A\|^2 \quad .
 \end{align}
We obtain
\begin{equation} 
X>\left(\frac{C_{\varphi 2}C_0^W}{d C_\varphi} \lambda^2\|\delta_A\|^2+\|\delta_\tau \|^2\right) \det A_2+O\left(\frac{\|\delta_\aff\|^2_\lambda}{\lambda^2}\right) \quad .
\end{equation}
We apply our estimates for $X$ and $Y$ to \eqref{Abstandteilgitterhilf1}, and get
\begin{align}
4J_{max}> &\left(\frac{C_0^WC_{\varphi 2}}{d C_\varphi} \lambda^2\|\delta_A\|^2+
C_0^W\|\delta_\tau\|^2\right)  \det A_R  \nonumber\\
& -C_0^W (\det A_R-\rho_\lambda(\chi^{reg}_{\aff_2},0) ) (8\lambda^2(\|\delta_A\|)^2+2\|\delta_\tau\|^2) \quad .
\end{align}
 We resubstitute $\delta_A$, and $\delta_\tau$ with equation \ref{regulardelta} and obtain 
\begin{align}
\|1-A^{-1}_1BA_2\|^2<&\left(\frac{C_0^WC_{\varphi 2}}{d C_\varphi} \det A_2- 8 C_0^W(\det A_2-\rho_\lambda(\chi^{reg}_{\aff_2},0))\right)^{-1}  \frac{4J_{max}}{\lambda^2} , \nonumber\\
\|A^{-1}_1(B\tau_2+t -\tau_1)\|^2<& \left(C_0^W(2\rho_\lambda(\chi^{reg}_{\aff_2},0))-\det A_2)\right)^{-1} 4J_{max} \quad .
\end{align}
From this follows the estimate \eqref{regularresult} for sufficiently small $\epsilon_{\rho}$ .
\end{proof}

\begin{lemma}\label{Sectpunktbewegenlambdaaendern}
For all $C_A>0$ there exists $\hat{\lambda}>0$ and $\hat{\epsilon}>0$ such that for all $\lambda>\hat{\lambda}$, 
$ \aff=(A,\tau)\in Gl_d(\R)\times \R^d$ with $\|A^{-1}\|<C_A$, and all $ x,y\in B_{2\lambda}(\Omega)$ it holds,
if $(x,\aff)$   is $(\epsilon_\rho, \epsilon_J, C_A)$-regular 
and $|x-y|<\lambda$, then $(y, \tilde{\aff})=(y,=(A,\tau+A(y-x))$ is $(\tilde{\epsilon}_\rho, \tilde{\epsilon}_J, C_A)$-regular 
using the smaller $\tilde{\lambda}=\lambda-|x-y|$ where
\begin{align}
J_{\tilde{\lambda}}(\aff,\chi,y)\le&\left(\frac{\lambda}{\tilde{\lambda}}\right)^d J_\lambda(\aff,\chi,x)\quad ,\nonumber\\
\tilde{\epsilon}_\rho=&\left(\frac{\lambda}{\tilde{\lambda}}\right)^{d/2} \frac{C+O(\lambda^{-1})}{\lambda}(1+\epsilon_\rho)\epsilon_J+\left(\frac{\lambda}{\tilde{\lambda}}\right)^{d} \epsilon_\rho \quad ,\nonumber\\
\tilde{\epsilon}_J=&\left(\frac{\lambda}{\tilde{\lambda}}\right)^d\frac{1+\tilde{\epsilon}_\rho}{1-\epsilon_\rho}\epsilon_J\quad .
\end{align}
\end{lemma}

\begin{proof}
We claim that for every atom $x_i\in \chi$ it holds
\begin{equation}
\varphi\left(\tilde{\lambda}^{-1}\left|x_i-y\right|\right)\leq \varphi\left(\lambda^{-1}\left|x_i-x\right)\right) \quad .
\end{equation}
If it holds $\|x_i-x\|\le \lambda$, we have $\varphi\left(\tilde{\lambda}^{-1}\left|x_i-y\right|\right)\leq 1=\varphi\left(\lambda^{-1}\left|x_i-x\right|\right)$, since $1$ is the maximum of $\varphi$.
\begin{figure}[ht]\label{BildJumping1}
\center
\includegraphics[scale=0.5]{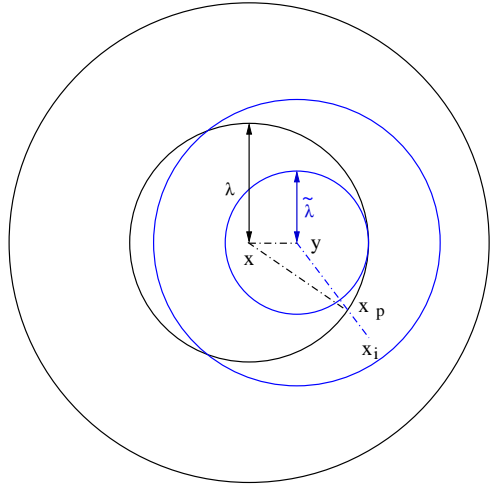}
\caption{Geometric setting}
\end{figure}
 $x_i$ is outside $B_\lambda(x)$ and $y$ is inside the ball.
line segment between $y$ and $x_i$ is intersecting with the surface of the ball in one point. We call this point $x_p$
(See picture \ref{BildJumping1}). We get
\begin{align}
\left|x-x_p\right|\le&\left|x-y\right|+\left|y-x_p\right| \quad ,\nonumber\\
\left|y-x_p\right|\ge&\left|x-x_p\right|-\left|x-y\right|=\lambda-\left|x-y\right|\ge\tilde{\lambda} \quad .
\end{align}
and
\begin{align}
\left|x_i-y\right|=&\left|x_i-x_p\right|+\left|x_p-y\right| 
         \ge\left|x_i-x_p\right|+\tilde{\lambda}\nonumber\\
         \ge&\frac{\tilde{\lambda}}{\lambda}\left|x_i-x_p\right|+\tilde{\lambda}
         \ge \frac{\left|x_i-x_p\right|+\lambda}{\lambda}\tilde{\lambda}
         \ge \frac{\left|x_i-x\right|}{\lambda}\tilde{\lambda}\quad .
\end{align}
Since $\varphi$ is monotone decreasing, we have
\begin{equation}\label{eqdifferentpointdensity}
\varphi\left(\tilde{\lambda}^{-1}\left|x_i-y\right|\right)\leq \varphi\left(\tilde{\lambda}^{-1}\left|x_i-x\right|\right) \quad .
\end{equation}
It holds
\begin{align}
J_{\tilde{\lambda}}(\tilde{\aff}, \chi, y)
=& \frac{\left\|A^{-1}\right\|^2}{C_{\varphi}\tilde{\lambda}^d}\sum_{i} W(A\left(x_i-y\right)+\tau+A(y-x))\varphi\left(\tilde{\lambda}^{-1}\left|x_i-y\right|\right)  \nonumber\\
\le&\frac{\lambda^d}{\tilde{\lambda}^d} \frac{\left\|A^{-1}\right\|^2}{C_{\varphi}\lambda^d}\sum_{i} W(A\left(x_i-x\right)+\tau)\varphi\left(\lambda^{-1}\left|x_i-x\right|\right) \nonumber\\
&\le\frac{\lambda^d}{\tilde{\lambda}^d} J_{\lambda}(\aff,\chi,x) \quad .
\end{align}
Now, we calculate a lower bound on $\rho_{\tilde{\lambda}}(\chi,y)$. 
We start at a Bravais lattice $\chi=\chi_\aff+x=A^{-1}(\Z^d-\tau)+x$ as a configuration. This configuration has $\epsilon_J=0$ 
For this lattice we have $\rho_{\tilde{\lambda}}(\chi,y)=\det A+O(\lambda^{-2})$.
There are different ways to reduce the density. On the one hand one can take atoms away. This decreases $\rho_{\tilde{\lambda}}(\chi,y)$ but because of equation \eqref{eqdifferentpointdensity} it also decreases $\rho_\lambda(\chi, x)$ at least by
\begin{equation}\label{GLlambdawechseldichte1}
\delta \rho_{\tilde{\lambda}}(\chi,y)< \frac{\lambda^d}{\tilde{\lambda}^d}\delta \rho_\lambda(\chi,x) \quad .
\end{equation}
Another possibility is to move atoms to positions of lower $\varphi\left(\lambda^{-1}\left|x_i-x\right|\right)$ this does not have to reduce $\rho_\lambda(\chi,x)$ at all but it will increase $J_\lambda(\aff,\chi,x)$.
If we shift the $i$`th atom for a distance $\delta x_i$ we maximally reduce  $\rho_{\tilde{\lambda}}(\chi,y)$ by 
\begin{equation}
 \delta \tilde{\rho}_i= \frac{1}{\lambda C_{\varphi}\tilde{\lambda}^d}\left|\nabla \varphi\right|\left(\lambda^{-1}\left|x_i-y\right|\right)\delta x_i +O(\lambda^{-1}\delta x_i) \quad ,
\end{equation} 
 we get a minimal cost per atom of 
\begin{equation}
\delta J_i <C_0^W \frac{1}{C_{\varphi}\tilde{\lambda}^d} \delta x_i^2 \varphi\left(\lambda^{-1}\left|x_i-x\right|\right)+
 O(\lambda^{-1}\delta x_i^2) \quad .
\end{equation} 
Furthermore, we have for $x_i\in B_{2\tilde{\lambda}}(y)$
\begin{align}
\left|x_i-x\right|\le & \left|x_i-y\right|+\left|y-x\right|\le 2 \tilde{\lambda}+\lambda-\tilde{\lambda}<2\lambda \quad .
\end{align}
Therefore, it holds $\varphi\left(\lambda^{-1}\left|x_i-x\right|\right)>0 $.
If we minimize $\sum_i\delta J_i$ with the constrain $\rho_\lambda = \det A+\sum_i \delta \tilde{\rho}_i$, we get
\begin{equation}
J_\lambda>\left(\frac{\tilde{\lambda}}{\lambda}\right)^d \left(\frac{1}{C_{\varphi}\tilde{\lambda}^d} \sum_{x_i\in B_{2\tilde{\lambda}}(y)}\frac{\left|\nabla \tilde{\varphi}\right|^2 \left(\lambda^{-1}\left(x_i-y\right)\right) }{\varphi\left(\lambda^{-1}\left| x_i-x\right|\right)}+O(\lambda^{-1})\right)^{-1}\lambda^2 \delta \rho^2 .
\end{equation}
We estimate the sum over $\chi_\aff$ by an integral using Lemma \ref{condistapp}. The error is $O(\lambda^{-2})$ that means  negligible  compared to the error we already made.
 \begin{align}\label{GLlambdawechseldichte2}
 J_\lambda \ge &\left(\frac{\tilde{\lambda}}{\lambda}\right)^d \left(\frac{1}{C_{\varphi}\tilde{\lambda}^d} \int_{\R^3} \frac{\left|\nabla \varphi (\tilde{\lambda}^{-1}\left|z-y\right|)\right|^2 }{\varphi\left(\lambda^{-1}\left|z-x\right|\right)}\det A dz+O(\lambda^{-1})\right)^{-1}\lambda^2 \delta \rho^2\quad ,\nonumber\\
 \delta \rho <&\left(\frac{\lambda}{\tilde{\lambda}}\right)^{d/2} \left(\frac{1}{C_{\varphi}} \int_{\R^3} \frac{\left|\nabla \varphi (\left|z-y\right|)\right|^2 }{\varphi\left(\left|z-x\right|\frac{\tilde{\lambda}}{\lambda} \right)} dz+O(\lambda^{-1})\right)^{\frac{1}{2}} \sqrt{\det A}  \frac{\sqrt{J_\lambda}(A,\tau, \chi,x)}{\lambda}\quad,\nonumber\\
  \delta \rho<&\tilde{C_\rho}\left(\frac{\lambda}{\tilde{\lambda}}\right)^{d/2} \frac{\sqrt{J_\lambda}(\aff, \chi,x)}{\lambda} \quad .
\end{align}
We can estimate the density $\rho_\lambda(\chi,x)$ from above with $\rho(\chi,x)<(1+\epsilon_\rho)\det A$
We summarize the estimates \eqref{GLlambdawechseldichte1} and \eqref{GLlambdawechseldichte2} and we obtain
\begin{equation}
\det A- \rho_{\tilde{\lambda}}(\chi,y)<\left(\left( \frac{\lambda}{\tilde{\lambda}}\right)^{d/2}\frac{C+O(\lambda^{-1})}{\lambda}(1+\epsilon_\rho)\epsilon_J+\left(\frac{\lambda}{\tilde{\lambda}}\right)^{d} \epsilon_\rho\right) \det A \quad.
\end{equation}
If we start with a Bravais lattice and increase the density $\rho_{\tilde{\lambda}}(\chi,y)$ by changing the configuration there are two different ways.
On the one hand can shift atoms to positions of higher $\varphi$.
This leads to the same increase of $J_{\lambda}$ as in the reduction case.
On the other hand one can add more atoms.
This leads to the same increase of $\rho_\lambda(\chi,x)$ additionally it will increase $J_\lambda$ because new atoms can not be placed in the minima because all minima are occupied.
We get the same estimate for the upper bound of $\rho_{\tilde{\lambda}}(\chi,y)$
\begin{equation}
\left|\det A- \rho_{\tilde{\lambda}}(\chi,y)\right|<\left(\left( \frac{\lambda}{\tilde{\lambda}}\right)^{d/2}\frac{C+O(\lambda^{-1})}{\lambda}(1+\epsilon_\rho)\epsilon_J+\left(\frac{\lambda}{\tilde{\lambda}}\right)^{d} \epsilon_\rho\right) \det A \quad .
\end{equation}
Finally, we estimate
\begin{align}
J_{\tilde{\lambda}}(A,\tau,\chi,y)\le&\left(\frac{\lambda}{\tilde{\lambda}}\right)^d J_\lambda(A,\tau,\chi,x)
\le \left(\frac{\lambda}{\tilde{\lambda}}\right)^d \epsilon_J \rho_\lambda(\chi,x)\quad,\nonumber\\
\le &\left(1+\tilde{\epsilon}_\rho\right) \left(\frac{\lambda}{\tilde{\lambda}}\right)^d\epsilon_J  \det A 
\le \left(\frac{\lambda}{\tilde{\lambda}}\right)^d\frac{1+\tilde{\epsilon}_\rho}{1-\epsilon_\rho}\epsilon_J \rho_{\tilde{\lambda}}(\chi,y)\quad.
\end{align}

\end{proof}

\begin{lemma}\label{Theoremjumping}
For all $C_A>s_o$ there exists
$ \hat{\lambda}\in \R$, $\epsilon_\rho>0$ and $\epsilon_J>0$ such that for all $\lambda>\hat{\lambda}$,
 $ \aff_j=(A_j,\tau_j)\in Gl_d(\R)\times \R^d$ and $y_j \in B_{2\lambda}(\Omega)$ with $j=1,2,3$ 
 the following holds.
If $(y_j, \aff_j)$ is $(\epsilon_\rho, \epsilon_J, C_A)$-regular  for $j=1,2,3$ and
$|y_j-y_k|\le \frac{3}{2}\lambda$ for $j,k=1,2,3$,
then there exists uniquely defined reparametrisations $B_{j,k}\in Gl_d(\Z)$, $t_{j,k}\in Z^d$ճuch that 
\begin{align}
\|id-A_j^{-1}B_{j,k}A_k\|<& \frac{c^{A}_J}{\sqrt{\det A_k}}  \left(\frac{2\lambda}{2\lambda-|y_j-y_k|}\right)^{d/2}\frac{\sqrt{J_{j,k}}}{\lambda},  \nonumber\\
\left|\delta \tau_{j,k} \right|<&   \frac{c^{\tau}_J|A_j|}{\sqrt{\det A_k}} \left(\frac{2\lambda}{2\lambda-\left|y_j-y_k\right|}\right)^{d/2} \sqrt{J_{j,k}},
\end{align}
where 
\begin{align}
   \delta \tau_{j,k} :=&B_{j,k}՜tau_k+t_{j,k} -\tau_j-\frac{B_{j,k}A_k+A_j}{2}\left(y_k-y_j\right)\quad ,\nonumber\\
J_{j,k}:=&\max\left\{ J_\lambda(\aff_j,\chi,y_j), J_\lambda(\aff_k,\chi,y_k)\right\}\quad ,\nonumber\\
c^{A}_J:=& \frac{3}{2} \left( \frac{8 d C_\varphi} {C_{\varphi 2}C_0^W}\right)^{\frac{1}{2}} \quad ,\quad c^{\tau}_J:= \left(\frac{10}{C_0^W}\right)^{\frac{1}{2}} \quad .
\end{align}
Moreover, it holds
\begin{align}\label{dreieckjumping}
B_{1,3}=&B_{1,2}B_{2,3}\quad ,\nonumber\\
t_{1,3}=&B_{1,2}t_{2,3}+t_{1,2}
\end{align}
\end{lemma}
\begin{proof}
We consider $\bar{y}=(y_j+y_k)/2$ and get
\begin{equation}
\left|y_j-\bar{y}\right|=\left|y_k-\bar{y}\right|=\left|y_j-y_k\right|/2<3/4 \lambda
\end{equation}
We apply Lemma \ref{Sectpunktbewegenlambdaaendern} twice, one time with $y_j$ as $x$ and $\bar{y}$ as $y$ and the other time with $y_k$ as $x$ and $\bar{y}$ as $y$. $(\bar{y}, \tilde{\aff}_j:=(A_j, \tau_j+A_j\left(x_k-x_j\right)/2)$ and $(\bar{y}, \tilde{\aff}_k:=(A_k, \tau_k+A_k\left(x_j-x_k\right)/2))$ are $\left(\tilde{\epsilon}_J, \tilde{\epsilon}_\rho , C_A \right)$-regular.  
 Therefore, we get
\begin{align}\label{boundJ}
J_\lambda(\tilde{\aff}_j,\chi,\bar{y})=&\left(\frac{2\lambda}{2\lambda-\left|y_j-y_k\right|}\right)^d J_{\tilde{\lambda}}(x, A, \tau+A(\bar{y}-x_j)) \quad ,\nonumber\\
\tilde{\epsilon}_\rho=&\left(\frac{2\lambda}{2\lambda-\left|y_j-y_k\right|}\right)^{d/2}\frac{C+O(\lambda^{-1})}{\lambda}(1+\epsilon_\rho)\epsilon_J+\left(\frac{2\lambda}{2\lambda-\left|y_j-y_k\right|}\right)^{d} \epsilon_\rho .
\end{align}
Since we have two regular pairs, we can apply Lemma \ref{Lemmatwodifferentfittedlattices} 
and get  $B_{j,k}\in \GL_d(\Z)$ and $t_{j,k}\in\Z^d$ such that
\begin{align}\label{boundaeinfach}
\|1-A_j^{-1}B_{j,k}A_k\|<& \frac{C^{A}_J}{\sqrt{\det A_k}}\left(\frac{2\lambda}{2\lambda-|y_j-y_k|}\right)^{d/2}  \frac{\sqrt{J_\lambda}}{\lambda} \quad ,\nonumber\\
\left|B_{j,k}\tau_k -\tau_j+\frac{B_{j,k}A_k+A_j}{2}\left(y_k-y_j\right)+t\right|<&  \frac{c^{\tau}_J\|A_j|}{\sqrt{\det A_k}}\left(\frac{2\lambda}{2\lambda-\left|y_j-y_k\right|}\right)^{d/2}  \sqrt{J_{j,k}}.
\end{align}
This proves the first part of the theorem.
Since it holds $|A_j^{-1}|\le C_{A}$, the matrix derivative of $\det A_j$ is bounded. Additionally, it holds 
$\det A_j<(1-\epsilon_\rho)^{-1}\rho^{max}_d $. 
Hence, the estimate \eqref{boundaeinfach} implies that we can estimate 
\begin{equation}
\det A_j=\det A_k+O\left(\frac{\sqrt{J_\lambda}}{\lambda}\right)\quad . 
\end{equation}
Due to the regularity condition on the density we get $\rho_\lambda (\chi, x_j)<(1+\epsilon_\rho) \det A_j$.
Therefore, we get for small enough $\epsilon_J$
\begin{align}
J_{j,k}&=\max\left\{ J_\lambda(\aff_{j},\chi,y_j), J_\lambda(\aff_k,\chi,y_k)\right\}\le \epsilon_J \max\left\{ \rho_\lambda(\chi,y_{j}), \rho_\lambda(\chi,y_k)\right\}\quad \nonumber\\
\le& \epsilon_J (1+\epsilon_\rho) \max\left\{ \det A_{}, \det A_k\right\}
\le\quad  \epsilon_J (1+\epsilon_\rho)\left(1+O\left(\frac{\sqrt{J_\lambda}}{\lambda}\right)\right)  \det A_k \quad .
\end{align}
Hence, for sufficiently large $\lambda$ we have 
\begin{align}\label{boundaeinfach2}
\|id-A_{j}^{-1}B_{j,k}A_k\|<&2 c^{A}_J 2^{d+1}\frac{\sqrt{\epsilon_J}}{\lambda}\quad .
\end{align}
We calculate for $\epsilon_J$ small enough
\begin{align}
\left|id-A^{-1}_{1}B_{1,2}B_{2,3}A_{3}\right|\le&\left|id-A^{-1}_{1}B_{1,2}A_2+A^{-1}_{1}B_{1,2}A_2\left(1-A_2^{-1}B_{2,3}A_{3}\right)\right|\nonumber\\
\le&\left|id-A_{1}^{-1}B_{1,2}A_2\right|+\left|A^{-1}_{1}B_{1,2}A_2\right|\left|id-A_1^{-1}B_{2,3}A_{3}\right|\nonumber\\
\le& 6c^{A}_J 2^d\frac{\sqrt{\epsilon_J}}{\lambda}   \quad .
\end{align}
 Due to the estimate \eqref{boundaeinfach2} we know 
\begin{align}
\|id-A_{1}^{-1}B_{1,3}A_{3}\|<2 c^{A}_J  2^d\frac{\sqrt{\epsilon_J}}{\lambda} \quad .
\end{align}
We get for $\epsilon_J<\lambda^2 (8C_A C_{|A|} C^{A}_J  2^d)^{-2}$
\begin{align}
\left|B_{1,3}- B_{1,2}B_{2,3}\right|\le&\left|A_{1}\right|\left|A_{1}^{-1}\left(B_{1,3}- B_{1,2}B_{2,3}\right)A_{3}\right|\left|A^{-1}_{3}\right|\nonumber\\
\le&C_A C_{|A|}\left|A_{1}^{-1}B_{1,3}A_3- A_{1}^{-1}B_{1,2}B_{2,3}A_{3}\right|\nonumber\\
\le&C_A C_{|A|}\left(\left|id-A_{1}^{-1}B_{1,3}A_3\right|+\left|id- A_{1}^{-1}B_{1,2}B_{2,3}A_{3}\right|\right)\nonumber\\
\le&8c^{A}_J C_A C_{|A|}2^d\frac{\sqrt{\epsilon_J}}{\lambda}<1 \quad .
\end{align}
The distance between $B_{1,3}$ and $B_{1,2}B_{2,3}$ is smaller than one and they are both elements of the discrete set $Gl_d(\Z)$.
Therefore, they have to be equal.
$B_{1,3}=B_{1,2}B_{2,3}$ also implies the uniqueness of $B_{1,3}$ and because $B_{1,2}$ and $B_{2,3}$ are invertible also the 
uniqueness of $B_{1,2}$ and $B_{2,3}$.
Next, we need  prove $t_{1,2}+B_{1,2}t_{2,3}=t_{1,3}$. 
If we apply estimate \eqref{boundJ} on \eqref{boundaeinfach} for the first chain, we get
\begin{align}\label{topologicaltau}
\delta \tau_{j,k}:=&B_{j,k}\tau_k+t_{j,k} -\tau_{j}-\frac{B_{j,k}A_{k}+A_{j}}{2}\left(y_k-y_{j}\right)\quad ,\nonumber\\
\left|\delta \tau_{j,k}\right|<&2 c^{\tau}_JC_{|A|}   2^d \sqrt{\epsilon_J} \quad.
\end{align}
Due to  the estimates \eqref{boundJ} and \eqref{boundaeinfach} it holds for $\epsilon_J<\lambda^2(c^{A}_J 2^d)^{-2}$
\begin{align}\label{topologicalA}
\left|B_{1,2}A_{2}-A_{1}\right|\le&\left|A_{1}^{-1}B_{1,2}A_{2}-id\right|\left|A_{2}\right|
\le 2 C_{|A|}c^{A}_J 2^d \frac{\sqrt{\epsilon_J}}{\lambda}\quad ,\nonumber\\
\left|B_{1,2}A_{2}-B_{1,2}B_{2,3}A_{3}\right|\le&\left|B_{1,2}A_2\right|\left|id-A_2^{-1}B_{2,3}A_{3}\right|
\le 2\left|A_{1}\right| \left|A_{1}^{-1}B_{1,2}A_2\right|c^{A}_J 2^d \frac{\sqrt{\epsilon_J}}{\lambda}\nonumber\\
\le&2C_{|A|} \left(1+2c^{A}_J 2^d \frac{\sqrt{\epsilon_J}}{\lambda}\right)    c^{A}_J 2^d \frac{\sqrt{\epsilon_J}}{\lambda}\le 4 C_{|A|}c^{A}_J 2^d \frac{\sqrt{\epsilon_J}}{\lambda} \quad .
\end{align}
We also get
\begin{equation}\label{topologicalB}
\left|B_{1,2}\right|\le\left|A_{1}\right|\left|A_{1}^{-1}B_{1,2}A_2\right|\left|A_2^{-1}\right|
\le C_A C_{|A|}\left(1+ 2c^{A}_J 2^d \frac{\sqrt{\epsilon_J}}{\lambda}\right)\le 2C_A C_{|A|} \quad.
\end{equation}
Hence, we can estimate
\begin{align}
&\left|B_{1,2}B_{2,3}\tau_{3}+B_{1,2} t_{2,3}+t_{1,2}-\tau_{1}-\frac{B_{1,2}B_{2,3}A_{3}+A_{1}}{2}\left(y_{3}-y_{1}\right)\right|\nonumber\\
=&\left| \frac{B_{1,2}A_{2}+A_{1}}{2}\left(y_2-y_{1}\right)+B_{1,2}\frac{B_{2,3}A_{3}+A_{2}}{2}\left(y_{3}-y_{2}\right)
-\frac{B_{1,2}B_{2,3}A_{3}+A_{1}}{2}\left(y_{3}-y_{1}\right)\right|\nonumber\\
&+\left|B_{1,2}\delta \tau_{2,3}+\delta \tau_{1,2}\right|\nonumber\\
\le&\left|B_{1,2}\delta \tau_{2,3}\right|+\left|\delta \tau_{1,2}\right|+\frac{1}{2}\left|\left(B_{1,2}A_{2}-A_{1}\right)\left(y_{3}-y_{2}\right)\right|+\frac{1}{2}\left|\left(B_{1,2}A_{2}-B_{1,2}B_{2,3}A_{2}\right)\left(y_{2}-y_{1}\right)\right|\nonumber\\
\le&\left|B_{1,2}\right|\left|\delta \tau_{2,3}\right|+\left|\delta \tau_{1,2}\right|+\frac{3\lambda}{4}\left|B_{1,2}A_{2}-A_{1}\right|+
\frac{3\lambda}{4}\left|B_{1,2}A_{2}-B_{1,2}B_{2,3}A_{3}\right| \quad.
\end{align}
We use the estimates \eqref{topologicaltau}, \eqref{topologicalA} and \eqref{topologicalB}.
\begin{align}\label{topologicaltau1}
&\left|B_{1,2}B_{2,3}\tau_{3}+B_{1,2}t_{2,3}+t_{1,2}-\tau_{1}-\frac{B_{1,2}B_{2,3}A_{3}+A_{1}}{2}\left(y_{3}-y_{1}\right)\right|\nonumber\\
\le&\left(\left(4C_A C_{|A|}+2\right)c^{\tau}_J+9/2 C^{A}_J\right) 2^d C_{|A|}\sqrt{\epsilon_J} \quad. 
\end{align}
On the other hand, if we apply  the estimate  to the second chain, we get
\begin{align}\label{topologicaltau2}
\left|B_{1,2}B_{2,3}\tau_{3}+t_{1,3}-\tau_{1}-\frac{B_{1,2}B_{2,3}A_{3}+A_{1}}{2}\left(y_{3}-y_{1}\right)\right|
<&2 c^{\tau}_J C_{|A|}  2^d \sqrt{\epsilon_J}\quad.
\end{align}
Hence, if we denote $X:=B_{1,2}B_{2,3}\tau_{3}-\tau_{1}-\frac{B_{1,2}B_{2,3}A_{3}+A_{1}}{2}\left(y_{3}-y_{1}\right)$,
 we finally get with the estimates \eqref{topologicaltau1} and \eqref{topologicaltau2}
\begin{align}
\left|B_{1,2}t_{2,3}+t_{1,2}-t_{1,3}\right|\le&\left|B_{1,2}t_{2,3}+t_{1,2}+X\right|+\left|X-t_{1,3}\right|\nonumber\\
\le&\left((4C_AC_{|A|}+2)c^{\tau}_J+9/2c^{A}_J\right) 2^dC_{|A|}\sqrt{\epsilon_J} \quad. 
\end{align}
For $\epsilon_J< \left(\left((4C_A c_{|A|}+2)c^{\tau}_J+9/2 c^{A}_J\right)2^dC_{|A|}\right)^{-2}$ 
the difference between $B_{1,2}t_{2,3}+t_{1,2}$ and $t_{2,3}$ is smaller than $1$ and since both belong to the discrete set $\Z^d$, they have to be equal.
As in the case of $B$ the equation $B_{1,2}t_{2,3}+t_{1,2}=t_{1,3}$ implies the uniqueness of $t_{j,k}$.
\end{proof}
 

\paragraph{Proof of Theorem \ref{addingatoms}}
\begin{proof}
For sufficiently large $\lambda$ and small $\epsilon_J$ the conditions Lemma \ref{Theoremjumping}
are fulfilled for any $j$, if we take $y_{j}$ as the first point in Lemma \ref{Theoremjumping} $y_{j+1}$ as the second and the third point. 
Hence, we get a sequence $\affb_{j,j+1}$ fulfilling equation \ref{equationjumping} for every $j$.
Furthermore, we can apply Lemma \ref{Theoremjumping} on the three points $y_{k-1}$, $y_{k+1}$ and $y_{k}$
From the first part of Lemma \ref{Theoremjumping} we get the existence of  a reparametrisation $\affb_{k-1,k+1}$ between $y_{k-1}$ and $y_{k+1}$.
Due to equation \eqref{dreieckjumping} Lemma \ref{Theoremjumping} we get:
\begin{align}
B_{k-1,k}B_{k,k+1}=&B_{k-1,k+1}\quad,\nonumber\\
t_{k-1,k}+B_{k-1,k}t_{k,k+1}=&t_{k-1,k+1}\quad.
\end{align}
Therefore, we get equation \eqref{eqtopologicalchain}.
\end{proof}

\section{Proof of Theorem  \ref{TheoremLowerboundoftheensitywiththelocalminimizersofJ} }
The next lemma shows that low energy points are regular.
\begin{lemma}\label{lowenergyregular}
If $H(\chi)<\infty$ and $\hat{h}_\lambda(\chi,x)\le \epsilon\le \frac{1}{4}  \min \left\{C_1^{El}\det(E)^2,C_2^{El}|E|^2,\vartheta \det E \right\}$, then there exists $\hat{\aff}\in Gl_d(\R^d)\times \R^d$ such that 
$\hat{h}_\lambda(\chi,x)=h_\lambda(\hat{\aff,\chi,x})$ and $(x,\hat{\aff}9$  is $(\epsilon_\rho, \epsilon_J, C_A)$-regular; 
where
\begin{equation}
C_A= \frac{3^{d-1}|E|^{d-1}}{2^{d-2}\det E}   \quad, \quad \epsilon_\rho=   2 \vartheta \frac{\epsilon}{\det E}  \quad, \quad\epsilon_J=     \frac{4 \epsilon }{\det E}     \quad .
\end{equation}
Additionally $\det \hat{A}\le \frac{3}{2}\det E $
\end{lemma}

\begin{proof}
For a configuration of finite energy the hard core condition $|x_i-x_j|>s_o$ is fulfilled for all atoms. 
Furthermore, for all $\aff\in Gl(\R)\times \R^d$ satisfying
\begin{equation}
\epsilon\ge h_{\lambda}\left(\aff, \chi x\right )=F\left(A\right)+J_{\lambda}\left(\aff, \chi, x\right) +\nu_{\lambda}\left(\chi,A,x\right)
\end{equation}
it holds due to the positivity of $F$, $J_{\lambda}$ and $\nu_{\lambda}$ 
\begin{equation}\label{lowenergyregularverteilt}
F\left(A\right)\le \epsilon\quad, \quad
J_{\lambda}\left(\aff, \chi, x\right)\le \epsilon \quad ,\quad
\nu_{\lambda}\left(\chi,A,x\right) \le\epsilon\quad ,
\end{equation}
Hence, we have 
\begin{equation}
 \epsilon \ge F(A)\ge C_1^{El}\left(\det(E)-\det(A))\right)^2+C_2^{El}\text{dist}^2\left(A, E\; SO_d \right)\quad.
\end{equation}
Therefore, for $\epsilon\le \frac{1}{4}\min \left\{C_1^{El}\det(E)^2,C_2^{El}|E|^2   \right\}$ we have
\begin{align}\label{lowenergyregulara}
\frac{1}{2}\det E  &\le \det(A)\le \frac{3}{2}\det E \quad ,\nonumber\\
\frac{1}{2}|E|   &\le |A|\le \frac{3}{2}|E|\quad . 
\end{align}
Additionally we have
\begin{align}\label{lowenergyregularainvers}
 |A^{-1}|\le&|A|^{d-1} \det A^{-1}\le \frac{3^{d-1}|E|^{d-1}}{2^{d-2}\det E} \quad ,\nonumber\\
|A^{-1}|\ge & \left(|A|\det A^{-1}\right)^{\frac{1}{d-1}}\ge   \left(\frac{|E|}{ 3 \det E}\right)^{\frac{1}{d-1}}\quad.
\end{align}
Because $h_{\lambda}\left(\cdot, \chi, x\right )$ is periodic in $\tau$, we can restrict $\tau$ to the compact set $[0,1]^d$. 
Hence, $Gl_d^\epsilon:=\left\{\aff\in Gl_d(\R)|h_\lambda(\aff,\chi,x)\le \epsilon , \tau\in [0,1]^d\right\}$ is a compact subset of $Gl_d(\R)\times\R^d$.
Since it holds $\hat{h}_\lambda (\chi,x)= \inf_{\aff}\left\{h_{\lambda}\left(\aff, \chi, x\right )  \right\}<\epsilon $, the set $Gl_d^\epsilon\times [0,1]^d$ is not empty.  Hence, the continuous function $h_{\lambda}\left(\cdot, \chi, x\right )$ attains a minimum  $\hat{\aff}=(\hat{A},\hat{\tau})$ on the compact set $Gl_d^\epsilon\times [0,1]^d$ that is per definition the global minimizer of  $h_{\lambda}\left(\cdot, \chi, x\right )$.   Due to the estimates \eqref{lowenergyregularverteilt} $\hat{\aff}$ satisfies
\begin{equation}\label{lowenergyregulardenisty}
 \vartheta \left|\det A-\rho_\lambda(\chi,x)\right|\le  \nu_{\lambda}\left(A,\chi,x\right)\le 2 \frac{\epsilon}{\det E} \det A \quad. 
\end{equation}
If we use the estimates \eqref{lowenergyregulardenisty} and \eqref{lowenergyregularverteilt}, we obtain for $\epsilon\le \frac{1}{4} \vartheta \det E $
\begin{equation}
J_\lambda(\aff,\chi,x)<\epsilon < \frac{\epsilon}{2^{-1}\det A-\vartheta^{-1}\epsilon}\rho_\lambda(\chi,x)\le    \frac{4 \epsilon }{\det E}\rho_\lambda(\chi,x)
\end{equation}
\end{proof}

$J_\lambda$ is locally convex in $\aff$ for regular pairs.

\begin{lemma}\label{hjlocalconvexity}
For all $C_A$ there exists  $\hat{\lambda}, \epsilon_J$ such that for all $\lambda>\hat{\lambda}$,  for all  $(C_A, \epsilon_\rho, \epsilon_J)$-regular  $(x,\aff)=(x,(A,\tau))\in \in B_{2\lambda}(\Omega) \times Gl_d(\R)\times \R^d$ and all test matrices $\affm=(M,\mu)\in \R^{d\times d}\times \R^d$ it holds
\begin{align}
\partial_\aff^2 J_\lambda(\aff,\chi, x)[\affm]\ge& C_{Con} \left\|A^{-1}\right\|^2 \rho_\lambda \|\affm\|_\lambda^2\quad ,
\end{align}
where $C_{con}$ is defined by
\begin{equation}
C_{con}:= c_\Theta^0 \min\left\{\frac{1}{12}  , \frac{c_\Theta^0 C_{\varphi}^2}{4\left(9+d\right)w_{d-1}^2 4^d}\frac{\rho_\lambda^2}{\det A^2} \right\}\quad .
\end{equation}
\end{lemma}

\begin{proof}
The second derivative $\partial_\aff^2 J_\lambda$ tested by $\affm=(M,\mu)\in \R^{d\times d}\times \R^d$ is given by
\begin{align}
\partial_\aff^2 J_\lambda(\aff,\chi,x)[\affm]=&\frac{\left\|A^{-1}\right\|^2}{C_{\varphi}\lambda^d}\sum_{i} \nabla^2W(A\left(x_i-x\right)+\tau)[M(x_i-x)+\mu]\varphi \nonumber\\
&+2\frac{\partial_\aff\left\|A^{-1}\right\|^2[M]}{C_{\varphi}\lambda^d}\sum_{i} \left<\nabla W(A\left(x_i-x\right)+\tau),M(x_i-x)+\mu\right>\varphi\nonumber\\
&+\frac{\partial_\aff^2\left\|A^{-1}\right\|^2[M]}{C_{\varphi}\lambda^d}\sum_{i} W(A\left(x_i-x\right)+\tau)\varphi \quad .
\end{align}
The two last terms lower are of order $O(\lambda^{-1})\|\affm\|_\lambda$. Furthermore, we can split the first sum into one sum over the regular atoms  $\chi_{\aff,\beta, x}^{reg}$ with $\beta=\min\left\{|A|^{-1}\Theta_W, s_o/3\right\}$ and one sum over the irregular atoms, and get
\begin{align}
C_{\varphi}\lambda^d\partial_\aff^2 J_\lambda(\aff,\chi,x)[\affm]=&\left\|A^{-1}\right\|^2\sum_{x_i\in \chi_{\aff,\beta, x}^{reg}} \nabla^2W(A\left(x_i-x\right)+\tau)[M(x_i-x)+\mu]\varphi \nonumber\\
&+\left\|A^{-1}\right\|^2\sum_{x_i\in \chi_{\aff,\beta, x}^{irr}} \nabla^2W(A\left(x_i-x\right)+\tau)[M(x_i-x)+\mu]\varphi \nonumber\\&-O(\lambda^{-1})\|\affm\|_\lambda^2  .
\end{align}
 On the one hand all the regular atoms satisfy
\begin{align}
dist(x_i,\chi_\aff+x)\le&\beta \quad ,\nonumber\\
dist (A(x_i-x)+\tau, \Z^d)\le & \beta |A| \le \Theta_W \quad ,\nonumber\\
c_\Theta^0(M(x_i-x)+\mu)^2\le&\nabla^2W(A\left(x_i-x\right)+\tau)[M(x_i-x)+\mu]\quad .
\end{align}
Since $W$ is two times differentiable and periodic, there is an upper bound for its second derivative, which we can use to bound the contribution of the irregular atoms. Hence, we get
\begin{align}
\partial_\aff^2 J_\lambda(\aff,\chi,x)[\affm]\ge&c_\Theta^0\frac{\left\|A^{-1}\right\|^2}{C_{\varphi}\lambda^d}\sum_{x_i\in \chi_{\aff,\beta, x}^{reg}} (M(x_i-x)+\mu)^2\varphi\left(\lambda^{-1}\left|x_i-x\right|\right)\nonumber\\
&-8\left\|A^{-1}\right\|^2\|\nabla^2W\|_\infty \rho_{\aff,\beta}^{irr}(x)\|\affm\|_\lambda^2-O(\lambda^{-1})\|\affm\|_\lambda^2\quad .
\end{align}
We define the average particle position by
\begin{equation}
\bar{x}:=\left(\rho_{\aff,\beta}^{reg}(x)\right)^{-1}\frac{1}{C_{\varphi}\lambda^d}\sum_{x_i\in \chi_{\aff,\beta, x}^{reg}}x_i\varphi\left(\lambda^{-1}\left|x_i-x\right|\right) \quad .
\end{equation}
Using this definition we get
\begin{align}
\partial_\aff^2 J_\lambda(\aff,\chi,x)[\affm]\ge&c_\Theta^0\frac{\left\|A^{-1}\right\|^2}{C_{\varphi}\lambda^d}\sum_{x_i\in \chi_{\aff,\beta, x}^{reg}} \left( \left(M(x_i-\bar{x})\right)^2+\left(M(\bar{x}-x)+\mu\right)^2\right)\varphi\nonumber\\
&-\left\|A^{-1}\right\|^2\|\nabla^2W\|_\infty \rho_{\aff,\beta}^{irr}(x)+O(\lambda^{-1})\|\affm\|_\lambda^2 \quad .
\end{align}
Because $(M(\bar{x}-x)+\mu)^2$ is independent of $i$, this sum can be expressed with the density of regular atoms.
If we denote by $e_M$ the eigenvector  the largest eigenvalue  of $M^TM$, we get
\begin{align}\label{convexpre1}
\partial_\aff^2 J_\lambda(\aff,\chi,x)[\affm]\ge&c_\Theta^0\frac{\left\|A^{-1}\right\|^2}{C_{\varphi}\lambda^d}\sum_{x_i\in \chi_{\aff,\beta, x}^{reg}} 
\left(e_m(x_i-\bar{x})\right)^2|M|^2\varphi \nonumber\\
&+c_\Theta^0\left\|A^{-1}\right\|^2 (M(\bar{x}-x)+\mu)^2\rho_{\aff,\beta}^{reg}(x)\nonumber\\
&-8\left\|A^{-1}\right\|^2\|\nabla^2W\|_\infty \rho_{\aff,\beta}^{irr}(x)\|\affm\|_\lambda^2+O(\lambda^{-1})\|\affm\|_\lambda^2 \quad .
\end{align}
We concentrate on the calculation of
\begin{equation}
X:=\frac{1}{C_{\varphi}\lambda^d}\sum_{x_i\in \chi_{\aff,\beta, x}^{reg}} 
\left(e_m(x_i-\bar{x})\right)^2\varphi\left(\lambda^{-1}\left|x_i-x\right|\right) \quad .
\end{equation}
Due to $\beta\le s_o/3$ there can be only one regular atom in $B_\beta(A^{-1}(z_i-\tau)+x)$ for any $z_i$
Therefore, the regular atoms can not sit all on the plain $P:=\{y\in R^d |e_m (y-\bar{x})=0\}$. 
We call $h$ the minimal distance to the plain $P$ up to which we have to fill atoms to reach the density $\rho_{\aff,\beta}^{reg}(x)$.
We define the cylinder
\begin{equation}
Z_P:=\left\{y\left| \left|\left<e_M,y-x\right>\right|\le 2\lambda \right.\right\} \quad .
\end{equation}
The characteristic function $1_{Z_P}$ of this set satisfies:
\begin{equation}
1_{Z_P}(x)\ge \varphi\left(\lambda^{-1}\left|x_i-x\right|\right) \quad.
\end{equation}
Hence, it holds
\begin{align}
C_{\varphi}\lambda^d\rho_{\aff,\beta}^{reg}(x)=&
\sum_{x_i\in \chi_{\aff,\beta, x}^{reg}} \varphi\left(\lambda^{-1}\left|x_i-x\right|\right)
\le  \sum_{x_i\in \chi_{\aff,\beta, x}^{reg}} 1_{Z_P}(x_i)
\le 2 w_{d-1} (2\lambda)^{d-1}\det A h \quad .
\end{align}
and we get
\begin{equation}
h\ge \frac{C_{\varphi} \lambda\rho_{\aff,\beta}^{reg}(x)}{w_{d-1}2^d\det A} \quad.
\end{equation}
Since for any valley with distance less then $h$ from the plain $P$, that does not have a regular atom, there needs to be an regular atom with larger distance to  reach the same density. Filling the whole cylinder gives us a lower bound for $X$
\begin{align}
X\ge& \frac{1}{C_{\varphi}\lambda^d}\int_0^h 2\tilde{h}^2  w_{d-1} (2\lambda)^{d-1}\det A d\tilde{h}
\ge \frac{C_{\varphi}^2}{3w_{d-1}^2 4^d}  \lambda^2 \det A^{-2}\left(\rho_{\aff,\beta}^{reg}\right)^3 \quad .
\end{align}
We apply this on the estimate \eqref{convexpre1} and get
\begin{align}\label{convexpre2}
\partial_\aff^2 J_\lambda(\aff,\chi,x)[\affm]
\ge&c_\Theta^0\left\|A^{-1}\right\|^2|M|^2 \frac{C_{\varphi}^2}{3w_{d-1}^24^d}  \lambda^2 \det A^{-2}\left(\rho_{\aff,\beta}^{reg}\right)^3\nonumber\\
&+c_\Theta^0\left\|A^{-1}\right\|^2 (M(\bar{x}-x)+\mu)^2\rho_{\aff,\beta}^{reg}(x)\nonumber\\
&-\left\|A^{-1}\right\|^2\|\nabla^2W\|_\infty \rho_{\aff,\beta}^{irr}(x)+O(\lambda^{-1})\|\affm\|_\lambda^2 \quad .
\end{align}
 We treat two cases. In case one it holds $|\mu|<3\lambda |M|$. In case two holds $|\mu|\ge 3 \lambda |M|$.
For case one we calculate
\begin{equation}
(9+d)\lambda^2|M|^2\ge d \lambda^2|M|^2+|\mu|^2= \lambda^2\|M\|^2+|\mu|^2=\|\affm\|_\lambda^2 \quad.
\end{equation}
We apply this to the estimate \eqref{convexpre2} and get
\begin{align}\label{convexpre3}
\partial_\aff^2 J_\lambda(\aff,\chi,x)[\affm]\ge&\frac{c_\Theta^0C_{\varphi}^2}{3\left(9+d\right)w_{d-1}^24^d}\left\|A^{-1}\right\|^2    \det A^{-2}\left(\rho_{\aff,\beta}^{reg}\right)^3 \|\affm\|_\lambda^2\nonumber\\
&-8\left\|A^{-1}\right\|^2\|\nabla^2W\|_\infty \rho_{\aff,\beta}^{irr}(x)+O(\lambda^{-1})\|\affm\|_\lambda^2\quad .
\end{align}
Since every atom contributing to the average $\bar{x}$ is in $B_{2\lambda}(x)$, also $\bar{x}$ itself has to be in $B_{2\lambda}(x)$. Therefore, we obtain for case two
\begin{equation}
(M(\bar{x}-x)+\mu)^2\ge \left(|\mu|-|M| |\bar{x}-x|\right)^2\ge \left(|\mu|-|M| 2\lambda\right)^2\ge \frac{1}{9} |\mu|^2 \quad.
\end{equation}
With estimate \eqref{convexpre2} we get
\begin{align}\label{convexpre4}
\partial_\aff^2 J_\lambda(\aff,\chi,x)[\affm]\ge&c_\Theta^0\left\|A^{-1}\right\|^2\|M\|^2 \frac{dC_{\varphi}^2}{3d w_{d-1}^24^d}  \lambda^2 \det A^{-2}\left(\rho_{\aff,\beta}^{reg}\right)^3 +\frac{c_\Theta^0}{9}\left\|A^{-1}\right\|^2 |\mu|^2 \rho_{\aff,\beta}^{reg}(x)\nonumber\\
&-\left\|A^{-1}\right\|^2\|\nabla^2W\|_\infty \rho_{\aff,\beta}^{irr}(x)+O(\lambda^{-1})\|\affm\|_\lambda^2\quad.
\end{align}
We summarize the inequalities \eqref{convexpre3} and \eqref{convexpre4} to get
\begin{align}\label{convexpre5}
\partial_\aff^2 J_\lambda(\aff,\chi,x)[\affm]\le& c_\Theta^0\left\|A^{-1}\right\|^2 \rho_{\aff,\beta}^{reg}(x) \alpha  \|\affm\|_\lambda^2-\left\|A^{-1}\right\|^2\|\nabla^2W\|_\infty \rho_{\aff,\beta}^{irr}(x)+O(\lambda^{-1})\|\affm\|_\lambda^2\quad.
\end{align}
where $\alpha$ is defined by
\begin{equation}
\alpha:= \min\left\{ \frac{1}{9}  , \frac{c_\Theta^0 C_{\varphi}^2}{3\left(9+d\right)w_{d-1}^2 4^d}   \frac{\left(\rho_{\aff,\beta}^{reg}\right)^2}{(\det A)^2} \right\}\quad .
\end{equation} 
We know from Lemma \ref{Dichteregular} with $\beta:=\min\{|A|^{-1}\Theta_W, s_o/3\}$ that it holds
\begin{align}
\rho_{\aff, \beta}^{irr}(x)\le&\frac{1}{C_0^W \min\{|A|^{-1}\Theta_W, s_o/3\}^2 } J_\lambda(\aff,\chi,x)\quad ,  \\
\rho_{\aff,\beta}^{reg}(x)\ge&\rho_\lambda(\chi,x)-\frac{1}{C_0^W \min\{|A|^{-1}\Theta_W, s_o/3\}^2 } J_\lambda(\aff,\chi,x) \quad . 
\end{align}
Therefore, we can control $\rho_{\aff ,\beta}^{irr}$ and $\rho_\lambda-\rho_{ \aff, \beta}^{reg}(x)$ for sufficiently low $\epsilon_J$ and large $\lambda$ arriving at
\begin{equation}
\partial_\aff^2 J_\lambda(\aff,\chi,x)[\affm]\ge \frac{7}{8}\alpha c_\Theta^0 \left\|A^{-1}\right\|^2 \rho_\lambda (x)   \|\affm\|_\lambda^2 \quad .
\end{equation}
\end{proof}


\begin{lemma}\label{estimatesumnablaW2}
 For all configurations $\chi$ and all $\aff\in Gl_d(\R)\times \R^d$ we have
\begin{equation}
J_\lambda(\aff,\chi,x)\ge\alpha_\nabla^{-1} \left\|A^{-1}\right\|^2C_{\varphi}^{-1}\lambda^{-d} \sum_i |\nabla W|^2(A\left(x_i-x\right)+\tau) \varphi\left(\lambda^{-1}|x-x_i|\right)  \quad ,
\end{equation}
where
\begin{equation}
 \alpha_\nabla:=64\max\left\{ \frac{ \|\nabla W\|_{\infty}^2}{C^W_0 \Theta_W^2  }, \frac{|c_\Theta^1|^2}{c_\Theta^0}\right\} \quad .
 \end{equation}
 \end{lemma}

\begin{proof}
We bound  $W(A\left(x_i-x\right)+\tau)$ from below with $(\nabla W)^2(A\left(x_i-x\right)+\tau)$.  
We define for every atom
\begin{equation}
\delta z_i:=\dist \left(A(x_i-x)+\tau,\Z^d\right) \quad .
\end{equation}
Due to the bounds on the second derivative of $W$ in the convex region we get for atoms with $\delta z_i\le \Theta_W$ 
\begin{equation}
(\nabla W)^2(\delta z_i)\le c_\Theta^1 |\delta z_i|^2\le \frac{2c_\Theta^1}{c_\Theta^0} W(\delta z_i) \quad .
\end{equation}
Due to the general bound $\|W\|_\infty$ we get for atoms with $\delta z_i\ge \Theta_W$ 
\begin{equation}
(\nabla W)^2(\delta z_i)\le \|\nabla W\|_{\infty}^2\le \frac{2 \|\nabla W\|_{\infty}^2}{C^W_0 \Theta_W^2  } W(\delta z_i) \quad .
\end{equation}
Hence, for the maximum $\alpha_\nabla:=64\max\left\{ \frac{ \|\nabla W\|_{\infty}^2}{C^W_0 \Theta_W^2  }, \frac{c_\Theta^1}{c_\Theta^0}\right\}$ we get for all atoms
\begin{align}\label{nablaWquadratkleinerW}
(\nabla W)^2(\delta z_i)\le& \alpha_\nabla W(\delta z_i)\quad ,\nonumber\\
J_\lambda(\aff,\chi,x)\ge&\alpha_\nabla^{-1} \left\|A^{-1}\right\|^2C_{\varphi}^{-1}\lambda^{-d} \sum_i |\nabla W|^2(A\left(x_i-x\right)+\tau) \varphi \quad .
\end{align}
\end{proof}


\begin{lemma}\label{localaffgradient}
For all $C_A$ there exists   $\hat{\lambda}>0$, $\epsilon_\rho>0$, $\epsilon_J>0$, $\delta_\aff>0$ such that for all $\lambda>\hat{\lambda}$,
$x_0 \in \Omega$ and $\aff_0\in Gl_d(\R)\times \R^d$
the following holds: If $(x_0\,\aff_0)$ is $ (C_A, \epsilon_\rho,\epsilon_J)$-regular
, then holds 
\begin{enumerate}[1)]
\item There exists a unique local minimizer of $J_\lambda$ 
\begin{equation}\label{Jlocalminimizerreichweite}
\tilde{\aff}=\argmin\{J_\lambda(\aff,\chi,x) |\aff\in Gl_d(\R)\times \R^d \text{with } \|\aff-\aff_0\|_\lambda<\delta_\aff \} \quad ,
\end{equation}
\item The local minimizer fulfills
\begin{align}
\left\|\aff_0-\tilde{\aff}\right\|_\lambda\le& \left(  \frac{1}{2} C_{Con} \|A_0^{-1}\|^2\rho_\lambda \right)^{-1/2}\sqrt{J_\lambda}(\aff_0,\chi,x)\quad,\label{Jlocalminimizerabstand}
\\ 
J_\lambda(\aff_0,\chi,x)\ge & J_\lambda(\tilde{\aff},\chi,x)+ \frac{1}{2} C_{Con} \left(\|A_0^{-1}\|^2+O(\lambda^{-1})\right) \rho_\lambda \left\|\aff_0-\tilde{\aff}\right\|_\lambda^2 \quad .\label{JlocalminimizerabstandJ}
\end{align}
\item For every differentiable curve $(x(s),\chi (s))$ with $x(0)=x$ and $\chi(0)=\chi$ there exists a neighborhood of $s=0$ such that inside this neighborhood $U$ there is a differentiable function $\tilde{\aff}: U\rightarrow Gl_d(\R)\times \R^d$ that $\aff(s)$  is a local minimizer of $J_\lambda(\cdot, \chi(s),x(s))$ for all $s\in U$ and fulfills
\begin{align}\label{Jlocalminimizergradient}
 \left\| \frac{d\tilde{\aff}(s)}{ds}  \right\|_\lambda \le&\frac{2\sqrt{8}|\tilde{A}(s)| }{ C_{Con}  \rho_\lambda } \left(\|\nabla^2 W\|_\infty+O(\lambda^{-1})\right)  \left(\sum_i \left|\frac{dx_i}{ds}(s)-\frac{dx}{ds}(s)\right| \varphi  \right) \nonumber\\ 
&+ O(\lambda^{-1}) \left(\sum_i \left(\left|\frac{dx_i}{ds}(s)-\frac{dx}{ds}(s)\right| |\nabla\tilde{\varphi}|  \right)\right) \quad .
\end{align}
\end{enumerate}
\end{lemma}

\begin{proof}
Since it holds $\|A_0^{-1}\|<C_A$ and the expressions $\|A^{-1}\|$, $|A|$ and $\det A$ are uniformly continuous functions of $A$ for regular points, we can find $\delta_{\aff}>0$ independent of $\lambda$ and $A$ such that for $\lambda\|A-A_0\|\le \delta_{\aff}$ holds
\begin{align}
\|A^{-1}\|<& C_A+O(\lambda^{-1})\quad ,\nonumber\\ 
\left|\rho_\lambda(\chi,x)-\det A\right|<&\left(\epsilon_{\rho}+O(\lambda^{-1})\right) \det A \quad .
\end{align}
Furthermore, we estimate 
\begin{align}
\left|\partial_\aff J_\lambda\left(\aff_0, \chi,x\right) [\affm]\right|\le & \left|\frac{\left\|A^{-1}_0\right\|^2}{C_{\varphi}\lambda^d}\sum_{i}\left< \nabla W(A_0\left(x_i-x\right)+\tau), M(x_i-x)+\mu\right>\varphi\right|+O(\frac{J_\lambda}{\lambda})\|\affm\|_\lambda .
\end{align}
We can use Cauchy-Schwarz inequality on the scalar product $\left<X,Y\right>_*:=\sum_{i}\left<X_i,Y_i\right>$ to get
\begin{align}
\left|\partial_\aff J_\lambda\left(\aff_0, \chi,x\right) [\affm]\right|
\le &\frac{\left\|A^{-1}_0\right\|^2}{C_{\varphi}\lambda^d}\left(\sum_{i} \left(\nabla W\right)^2 \varphi\right)^{\frac{1}{2}}
\left(\sum_{i} \left( M(x_i-x)+\mu\right)^2 \varphi\right)^{\frac{1}{2}}\nonumber\\
&+O(\lambda^{-1})\|\affm\|_\lambda J_\lambda\left(\aff_0, \chi,x\right)\quad .
\end{align}
Due to Lemma \ref{estimatesumnablaW2} we obtain the bound:
\begin{align}\label{exitslinearJ}
\left|\partial_\aff J_\lambda\left(\aff_0, \chi,x\right) [\affm]\right|
&\le O(\sqrt{J_\lambda}\left(\aff_0, \chi,x\right) \|\affm\|_\lambda)\quad .
\end{align}
Therefore, if we choose  $\tilde{\epsilon}_J$ fulfilling the conditions of Lemma \ref{hjlocalconvexity}, then for sufficiently small $\epsilon_J$  exists  $\delta_\aff>0$ such that $(x,\aff)$ is $(C_A+O(\lambda^{-1}), \epsilon_\rho+O(\lambda^{-1}), \tilde{\epsilon}_J)$-regular for $\|\aff-\aff_0\|_\lambda\le \delta_\aff$.
Hence, for sufficiently small $\epsilon_J$ all the conditions of Lemma \ref{hjlocalconvexity}  are satisfied. Furthermore, 
 $J_\lambda(\tilde{\aff}, \chi, x)$ is for $\|\aff-\aff_0\|_\lambda\le \delta_\aff$ a strictly convex function of $\aff$
\begin{equation}\label{exitsquadraticJ}
\partial^2_{\aff} J_\lambda( \aff, \chi, x)[\affm] \ge C_{Con} \|A^{-1}\|^2 \rho_\lambda \|\affm\|_\lambda^2 \quad .
\end{equation}
Hence, for any $\aff$ with $J_\lambda( \aff, \chi, x)\le J_\lambda( \aff_0, \chi, x)$ we consider $\bar{\aff}:=\frac{\aff_0+\aff}{2}$ as a starting point for a Taylor expansion.
\begin{align}
J_\lambda\left(\aff,\chi,x\right)\ge & J_\lambda\left(\frac{\aff_0+\aff}{2},\chi,x\right)+\frac{1}{2}\partial_\aff+J_\lambda\left(\frac{\aff_0+\aff}{2},\chi,x\right) \left[\aff-\aff_0\right]\nonumber\\ &+\frac{1}{9}C_{Con} \|A_0^{-1}\|^2 \rho_\lambda \|\aff_0-\aff \|_\lambda^2,\nonumber\\
J_\lambda\left(\aff_0,\chi,x\right)\ge &J_\lambda\left(\frac{\aff_0+\aff}{2},\chi,x\right)+\frac{1}{2}\partial_\aff+J_\lambda\left(\frac{\aff_0+\aff}{2},\chi,x\right) \left[\aff_0-\aff\right]\nonumber\\ &+\frac{1}{9}C_{Con} \|A_0^{-1}\|^2 \rho_\lambda \|\aff_0-\aff \|_\lambda^2.
\end{align}
If we add these estimates and apply $J_\lambda\left(\frac{\aff_0+\aff}{2},\chi,x\right)>0$ and $J_\lambda( \aff, \chi, x)\le J_\lambda( \aff_0, \chi, x)$, we get
\begin{equation}
\|\aff_0-\aff \|_\lambda^2\le R_\aff^2=: \frac{9 J_\lambda(\aff_0,\chi,x)}{ \|A_0^{-1}\|^2 \rho_\lambda}
\end{equation}
Hence, all $\aff$ with $J_\lambda( \aff, \chi, x)\le J_\lambda( \aff_0, \chi, x)$ are in the ball $B_{R_\aff}(\aff_0)$ with $R_\aff \le \delta_\aff $. 
Therefore the continous function $J_\lambda$ attains a minimum inside the ball $B_{R_\aff}(\aff_0)$ and therefore has a local minimum $\tilde{\aff}$ in $B_{R_\aff}(\aff_0)$. The local minimum fulfills $\partial_\aff J_\lambda(\tilde{\aff},\chi,x)=0$. Therefore, it holds
\begin{equation}
J_\lambda(\aff,\chi,x)\ge  J_\lambda(\tilde{\aff},\chi,x)+\frac{1}{2}C_{Con} \|A_0^{-1}\|^2 \rho_\lambda \|\aff-\tilde{\aff} \|_\lambda^2
\end{equation}
Hence, the minimizer is unique and we get the estimate \eqref{Jlocalminimizerabstand}.
Now we search as solution $\tilde{\aff}(s)$ for the equation
\begin{equation} \label{generalgrund}
0=\partial_{\aff}J_\lambda(\tilde{\aff},\chi(s),x(s)) \quad.
\end{equation}
According to implicit function theorem there is a differentiable solution $\tilde{\aff}(s)$ satisfying the equation \eqref{generalgrund},
if $\det \partial_{\aff}^2J_\lambda(\tilde{\aff},\chi(s),x(s)\ne 0$. This is implied by the strict convexity given by the estimate \eqref{exitsquadraticJ}. 
 Therefore, there exists a solution of the equation in this neighborhood and the solution is a local minimizers of $J_\lambda$.
Since $0=\partial_{\aff}J_\lambda(\tilde{\aff}(s),\chi(s),x(s)) $, it is also zero tested with any $\affm=(M,\mu) \in \R^{d\times d}\times \R$.
In particular $\partial_{\aff}J_\lambda(\tilde{\aff},\chi(s),x(s))[\affm]$  is constant and its derivative is zero.
We leave out the argument of $J_\lambda$ for simplicity and get
\begin{align}\label{existsimplizite1}
0=&\frac{d}{ds}\partial_{\aff}J_\lambda \left[\affm\right]
=\partial_{\aff}\left(\partial_{\aff}J_\lambda \left[\affm\right]\right)
\left[\frac{d\tilde{\aff}}{ds}(s)\right]+\partial_x \partial_{\aff}J_\lambda \left[\affm\right]\frac{dx(s)}{ds}+\partial_\chi \partial_{\aff}J_\lambda\left[\affm\right]\frac{d\chi(S)}{ds}.
\end{align}
If we test the estimate \eqref{existsimplizite1} with $\affm=\frac{d\tilde{\aff}(s)}{ds}$ and apply the local convexity from the inequality \eqref{exitsquadraticJ} we obtain
\begin{align}\label{existsimplizite2}
\frac{1}{2}C_{Con} \|A_0^{-1}\|^2 \rho_\lambda \| \frac{d\tilde{\aff}(s)}{ds}  \|_\lambda^2
\le &-\partial_x \partial_{\aff}J_\lambda\left(\tilde{\aff}(s),\chi(s),x(s)\right)\left[\frac{d\tilde{\aff}(s)}{ds}\right]\frac{dx(S)}{ds}\nonumber\\
&-\partial_\chi \partial_{\aff}J_\lambda\left(\tilde{\aff}(s),\chi(s),x(s)\right)\left[\frac{d\tilde{\aff}(s)}{ds}\right]\frac{d\chi(S)}{ds}\quad.
\end{align}
Furthermore, we estimate 
\begin{align} 
&C_{\varphi}\lambda^d\partial_x \partial_{\aff}J_\lambda\left(\tilde{\aff}(s),\chi(s),x(s)\right)\left[\affm\right]\frac{dx(S)}{ds}
+C_{\varphi}\lambda^d\partial_\chi \partial_{\aff}J_\lambda\left(\tilde{\aff}(s),\chi(s),x(s)\right)\left[\affm\right]\frac{d\chi(S)}{ds}\nonumber\\
=&-\left\|\tilde{A}^{-1}\right\|^2\sum_{i} \left<  \nabla^2 W \left(M(x_i-x)+\mu\right),\tilde{A} \left(\frac{dx_i}{ds}-\frac{dx}{ds}\right)    \right>  \varphi\nonumber\\
&-\left\|\tilde{A}^{-1}\right\|^2\sum_{i} \left<\nabla W, M \left(\frac{dx_i}{ds}-\frac{dx}{ds}\right)\right>     \varphi\nonumber\\
&+\partial_A\left\|\tilde{A}^{-1}\right\|^2[M]\sum_{i} \left<\nabla W, A\left(\frac{dx_i}{ds}-\frac{dx}{ds}\right)\right> \varphi\nonumber\\
&+\lambda^{-1}\left\|\tilde{A}^{-1}\right\|^2\sum_{i}  \left<\nabla W ,M(x_i-x)+\mu\right> \left<\nabla \tilde{\varphi},\left(\frac{dx_i}{ds}-\frac{dx}{ds}\right)\right>\nonumber\\
&+\lambda^{-1}\partial_A\left\|\tilde{A}^{-1}\right\|^2[M]\sum_{i}  W \left<\nabla \tilde{\varphi},\left(\frac{dx_i}{ds}-\frac{dx}{ds}\right)\right> \quad.
\end{align}
Hence, we get
\begin{align} \label{linearexistsimplicite}
&\left|C_{\varphi}\lambda^d\partial_x \partial_{\aff}J_\lambda\left(\tilde{\aff}(s),\chi(s),x(s)\right)\left[\affm\right]\frac{dx(s)}{ds}
+C_{\varphi}\lambda^d\partial_\chi \partial_{\aff}J_\lambda\left(\tilde{\aff}(s),\chi(s),x(s)\right)\left[\affm\right]\frac{d\chi(s)}{ds}\right|\nonumber\\
\le &  \sqrt{8}\left\|\tilde{A}^{-1}\right\|^2|\tilde{A}|  \left(\|\nabla^2 W\|_\infty+O(\lambda^{-1})\right)  
\|\affm\|_\lambda \left(\sum_i \left|\frac{dx_i}{ds}-\frac{dx}{ds}\right| \varphi  \right) \nonumber\\ 
&+ O(\lambda^{-1}) \|\affm\|_\lambda  \left(\sum_i \left(\left|\frac{dx_i}{ds}-\frac{dx}{ds}\right| |\nabla\tilde{\varphi}| \right)\right) \quad. 
\end{align}
If we combine the estimates \eqref{existsimplizite2} and \eqref{linearexistsimplicite}, we get
\begin{align}
\frac{1}{2}C_{Con} \|A_0^{-1}\|^2 \rho_\lambda \left\| \frac{d\tilde{\aff}(s)}{ds}  \right\|_\lambda^2\le& \sqrt{8}\left\|\tilde{A}^{-1}\right\|^2|\tilde{A}|  \|\nabla^2 W\|_\infty  
\left\|  \frac{d\tilde{\aff}(s)}{ds}   \right\|_\lambda \left(\sum_i \left|\frac{dx_i}{ds}-\frac{dx}{ds}\right| \varphi  \right) \nonumber\\ 
&+ O(\lambda^{-1}) \left\|\frac{d\tilde{\aff}(s)}{ds}   \right\|_\lambda  \left(\sum_i \left(\left|\frac{dx_i}{ds}-\frac{dx}{ds}\right| |\nabla\tilde{\varphi} |\right)\right) 
\end{align}
\end{proof}
Next, we improve the estimate for the gradients of the local minimizers. The basic idea is that $\nabla \tau$ has to be very similar to $A$.
Hence, if we do not estimate $\|\nabla \tau\|$ but $\|\nabla \tau-A\|$, we can get a much better estimate.
The result is similar to the final estimate in Theorem 4.5 from \cite{Lucapaper}.
However we improve the estimate so that we can use the gradient of the local minimizers to bound $J_\lambda$ from below. 
Furthermore, we use the same technique to get an estimate for the second gradient of the local minimizer $\aff$.

\begin{lemma}\label{Upperboundgradients}
For all $C_A$ and $\epsilon_\rho>0$ there exists  $\hat{\lambda}>0$, $\epsilon_J>0$, $\delta_\aff>0$ such that for all $\lambda>\hat{\lambda}$,
$x_0 \in \Omega$ and $\aff_0\in Gl_d(\R)\times \R^d$
the following holds: If $(x_0,\aff_0)$ is $ (C_A, \epsilon_\rho,\epsilon_J)$-regular,
then the gradients of the local minimizers $\tilde{\aff}$ (see Theorem \ref{localaffgradient})  satisfy
\begin{align}\label{upperboundgradientaff}
J_\lambda\left(\tilde{\aff}, \chi, x\right)\ge & \frac{C_{con}^2 \left\|\tilde{A}^{-1}\right\|^2}{\alpha_\nabla 2^d\|\nabla\sqrt{\tilde{\varphi}}\|_\infty^2} \frac{\rho_\lambda^2}{\rho_{2\lambda}} \lambda^2\left( \lambda^2\|\nabla \tilde{A}\|^2+\|\nabla \tilde{\tau}- \tilde{A}  \|^2\right) \quad. 
\end{align}
Furthermore, if $W$ is three times differentiable, we get:
\begin{equation}
J_\lambda\left(\tilde{\aff}, \chi, x\right)\ge C_{\nabla 2}\left(\frac{\rho_{2\lambda}}{\rho_\lambda}\right) \left\|\tilde{A}^{-1}\right\|^2\rho_\lambda \lambda^4 \left(\lambda^2\|\nabla^2 \tilde{A}\|^2+ \|\nabla^2\tilde{\tau}-\nabla \tilde{A}\|^2\right)\quad ,
\end{equation}
where 
\begin{align}
\left(C_{\nabla 2}(X)\right)^{-\frac{1}{2}}:=& \frac{\sqrt{\alpha_\nabla}}{C_{con} }\left(\|\nabla\sqrt{\tilde{\varphi}}\|_\infty^2+ \|\nabla^2\sqrt{\tilde{\varphi}}\|_\infty+ 2 \infty \|\nabla \sqrt[4]{\tilde{\varphi}}\|^2\right)d\sqrt{2^d X}\nonumber\\
& +\frac{\sqrt{\alpha_\nabla}} {C_{con}^2 }\left(2^d\|\nabla\sqrt{\tilde{\varphi}}\|_\infty^2\right)^{\frac{1}{2}}\left(16\;  2^{\frac{d}{2}} X +\sqrt{8d}\sqrt{X}\right)\quad.
\end{align}
\end{lemma}

\begin{proof}
{\bf Step 1: The first derivative:}
Since the same conditions as in Theorem \ref{localaffgradient} are fulfilled, we get the minimizer $\tilde{\aff}$.
The local minimizer fulfills $0=\partial_{\aff}J_\lambda(\tilde{\aff}, \chi,x)$. 
On the one hand this implies for the $\tau$ derivative
\begin{equation}\label{gradienttauzero}
0=\sum_{i} \nabla W(\tilde{A}\left(x_i-x\right)+\tilde{\tau})\varphi\left(\lambda^{-1}\left|x_i-x\right|\right) \quad .
\end{equation}
On the other hand, the total derivative of $\partial_{\aff}J_\lambda(\tilde{\aff}(x), \chi , x)$ in every direction  $e_j$ is zero, because we know that $\partial_{\aff}J_\lambda(\tilde{\aff}(x),\chi, x)[\affm]$ is constant.
In particular for all test matrices $\affm=(M,\mu)\in \R^{d\times d}\times \R^d$ holds
\begin{align}\label{implizitaffbais}
0=&\frac{d}{dx^j}\left(\partial_{\aff}J_\lambda(\tilde{\aff}(x),\chi , x)[\affm]\right)\nonumber\\
=&\partial^2_{\aff}J_\lambda(\tilde{\aff}(x), \chi , x)\left(\affm,\nabla_j \tilde{\aff}(x)\right) +\nabla_j \left(\partial_{\aff}J_\lambda(\tilde{\aff}(x),\chi , x)[\affm]\right)\quad .
\end{align}
We compare $\partial_{\tau_j}\left(\partial_{\aff}J_{\lambda}[\affm]\right)$ with $\nabla_j\partial_{\aff}J_{\lambda}[\affm]$. First, we calculate the $\tau$-derivative and then use equation \eqref{gradienttauzero}.
\begin{align}\label{ableitungafftau}
C_{\varphi}\lambda^d\partial_{\tau_j}\partial_{\aff} J_{\lambda}(\tilde{\aff},\chi, x)[\affm]=&\left\|\tilde{A}^{-1}\right\|^2\sum_{i}  \left<\nabla_j\nabla W(\tilde{A}\left(x_i-x\right)+\tilde{\tau}),M(x_i-x)+\mu\right> \varphi \quad.
\end{align}
Now, we calculate the partial derivative $\nabla_j\partial_{\aff} J_{\lambda}(\tilde{\aff}, \chi ,x)[\affm]$ and get 
\begin{align} \label{ableitungaffx}
&C_{\varphi}\lambda^d\nabla_j \partial_{\aff} J_{\lambda}(\tilde{\aff}(x), \chi ,x)[\affm]\nonumber\\
=&-\left\|\tilde{A}^{-1}\right\|^2\sum_{i} \left(\left<  \nabla^2 W \left(M(x_i-x)+\mu\right),\tilde{A}e_j\right>+\left<\nabla W, M e_j\right> \right)  \varphi\nonumber\\
&+\partial_A\left\|\tilde{A}^{-1}\right\|^2[M]\sum_{i} \left(\left<\nabla W, \tilde{A}e_j\right> \varphi +W \nabla_j \tilde{\varphi}   \right)+\left\|\tilde{A}^{-1}\right\|^2\sum_{i}  \left<\nabla W ,M(x_i-x)+\mu\right> \nabla_j \tilde{\varphi}.
\end{align}
The second and the third term are zero due to equation \eqref{gradienttauzero}. 
We compare the equations \eqref{ableitungafftau} and \eqref{ableitungaffx} and see that the first term of
$\nabla_j\left(\partial_{\aff} J_{\lambda}[\affm]\right)$ equals $-\left<\partial_x\left(\partial_{\aff} J_{\lambda}[\affm]\right),\tilde{A} e_j\right>$.
We summarize the last two terms into a linear map  $D:\R^{d\times d}\times \R^d \rightarrow \R^d$.
\begin{equation} \label{ableitungaffx2}
\nabla_j\partial_{\aff} J_{\lambda}(\tilde{\aff}(x), \chi, x)[\affm]
=-\left<\partial_{\tau}\partial_{\aff} J_{\lambda}(\tilde{\aff}(x),\chi, x)[\affm],\tilde{A}e_j\right>-D_j[\affm] \quad , 
\end{equation}
where $D[\affm]$ is defined by
\begin{align}
D[\affm]:=&-\frac{\left\|\tilde{A}^{-1}\right\|^2}{\lambda C_{\varphi}\lambda^d}\sum_{i} \left<\nabla W(\tilde{A}\left(x_i-x\right)+\tilde{\tau}),M(x_i-x)+\mu\right> \nabla\tilde{\varphi}\nonumber\\
&-\frac{\partial_A\left\|\tilde{A}^{-1}\right\|^2[M]}{\lambda C_{\varphi}\lambda^d}\sum_{i}  W(\tilde{A}\left(x_i-x\right)+\tilde{\tau}) \nabla \tilde{\varphi} \quad.
\end{align}
Using equation \eqref{ableitungaffx2} we can reformulate equation \eqref{implizitaffbais} as follows
\begin{equation}\label{implizitaffbais2}
D_j[\affm]=\partial_{\aff}^2J_\lambda(\tilde{\aff}(x), \chi, x)\left(\affm,\left(\nabla_j \tilde{A}, \nabla_j \tilde{\tau}-\tilde{A}e_j\right)\right)\quad .
\end{equation}
We test this equation with $\affm=(\nabla_j \tilde{A}, \nabla_j \tilde{\tau}-\tilde{A}e_j)$ and sum over $j$ to obtain
\begin{equation}\label{DistkleineralsquadratischeForm}
\sum_j D_j \left[\left(\nabla \tilde{A}, \nabla_j \tilde{\tau}-\tilde{A}e_j\right)\right]=\sum_j \partial_{\aff}^2J_\lambda(\tilde{\aff}(x), \chi, x)\left[\left(\nabla_j \tilde{A}, \nabla_j \tilde{\tau}-\tilde{A}e_j\right)\right]\quad.
\end{equation}
Using the local convexity proofed in lemma \ref{hjlocalconvexity}, we get:
\begin{align}\label{DistkleineralsquadratischeForm2}
\sum_j D_j\left(\nabla_j \tilde{A},  \nabla_j\tilde{\tau}-\tilde{A}e_j\right)\ge&  C_{con} \rho_\lambda \|\tilde{A}^{-1}\|^2 \left( \lambda^2\|\nabla \tilde{A}(x)\|^2+\|\nabla \tilde{\tau}(x)- \tilde{A}  \|^2\right) \quad .
\end{align}
We rewrite the left side of the last inequality 
\begin{align}\label{EstimateD1}
 C_{\varphi}\lambda^d\sum_j D _j\left[\left(\nabla_j \tilde{A}, \nabla_j \tilde{\tau}-\tilde{A} e_j\right)\right]=&-\lambda^{-1}
\left\|\tilde{A}^{-1}\right\|^2\sum_{i,j} \left<\nabla W ,\nabla_j \tilde{A}(x_i-x)+\nabla_j \tilde{\tau} \right> \nabla_j\tilde{\varphi}\nonumber\\&-\lambda^{-1}\sum_{i,j}\partial_A\left\|\tilde{A}^{-1}\right\|^2[\nabla_j \tilde{A}]  W \nabla_j\tilde{\varphi}  .
\end{align}
Moreover, we have
\begin{align}
\nabla \tilde{\varphi}(x)=&2 \sqrt{\tilde{\varphi}}(x)\nabla\sqrt{\tilde{\varphi}}(x) \quad ,\nonumber\\
\left(\nabla_j\tilde{A}(x_i-x)+\nabla_j\tilde{\tau}-\tilde{A}e_j\right)^2\le& 2|\nabla_j\tilde{A}|^2|x_i-x|^2+2|\nabla_j\tilde{\tau}-\tilde{A}e_j|^2\nonumber\\
\le& 8\lambda^2\|\nabla_j\tilde{A}\|^2+2|\nabla_j\tilde{\tau}-\tilde{A}e_j|^2 \quad.
\end{align}
Therefore, we use Cauchy-Schwarz inequality to estimate
\begin{align}\label{Cauchyschwarz}
& C_{\varphi}\lambda^{d+1}\sum_j D _j\left(\nabla_j \tilde{A}, \nabla_j \tilde{\tau}-\tilde{A}e_j\right)\nonumber\\
=&-2\sum_{i,j}\left(\left\|\tilde{A}^{-1}\right\|^2 \left<\nabla W, \nabla_j \tilde{A}(x_i-x)+\nabla_j \tilde{\tau}- \tilde{A}_{\cdot j} \right>+
\partial_A\left\|\tilde{A}^{-1}\right\|^2[\nabla_j \tilde{A}] W  \right) \sqrt{\tilde{\varphi}} \nabla_j\sqrt{\tilde{\varphi}}\nonumber\\
=&2 \left\|\tilde{A}^{-1}\right\|^2\sum_j\left( \sum_i (\nabla W)^2  \left(8\lambda^2\|\nabla \tilde{A}_j\|^2+ 2|\nabla_j \tilde{\tau}-\tilde{A}e_j|^2 \right)\tilde{\varphi}^2\right)^{\frac{1}{2}} \left(\sum_{i}|\nabla \sqrt{\tilde{\varphi}}|^2\right)^{\frac{1}{2}}\nonumber\\
&+2 \left|\partial_A\left\|\tilde{A}^{-1}\right\|^2\right|\|\nabla \tilde{A}\|  \left(\sum_i  W^2  \tilde{\varphi}\right)^{\frac{1}{2}}\left(\sum_{i}|\nabla\sqrt{\tilde{\varphi}}|^2\right)^{\frac{1}{2}} \quad.
\end{align}
Since we have $\varphi(z)=1$ for $z\le 1$, we estimate 
\begin{align}
\sum_{i}|\nabla\sqrt{\tilde{\varphi}}|^2\left(\lambda^{-1}\left|x_i-x\right|\right)=&  \sum_i |\nabla\sqrt{\tilde{\varphi}}|^2\left(\lambda^{-1}\left|x_i-x\right|\right)\varphi\left((2\lambda)^{-1}\left|x_i-x\right|\right)\nonumber\\
\le&C_{\varphi}(2\lambda)^d\|\nabla \sqrt{\tilde{\varphi}}\|_\infty^2 \rho_{2\lambda}(\chi, x) \quad .
\end{align}
We use $W^2(\tilde{A}\left(x_i-x\right)+\tilde{\tau})\le \|W\|_\infty W(\tilde{A}\left(x_i-x\right)+\tilde{\tau})$
and Lemma \ref{estimatesumnablaW2} on the inequality \eqref{Cauchyschwarz} and get
\begin{align}
\sum_j D _j\left(\nabla_j \tilde{A}, \nabla_j \tilde{\tau}-\tilde{A}\right)\le& \left(\left(\alpha_\nabla\right)^{\frac{1}{2}}\left\|\tilde{A}^{-1}\right\| 
 +O\left(\lambda^{-1}\right)\right) \left(2^d\|\nabla \sqrt{\tilde{\varphi}}\|_\infty^2 \rho_{2\lambda}\right)^{\frac{1}{2}}\nonumber\\
 &\times \lambda^{-1}\sqrt{J}_\lambda\left(\tilde{\aff}, \chi, x\right) \left(\lambda^2|\nabla \tilde{A}\|^2+\|\nabla \tilde{\tau}-\tilde{A}\|^2\right)^{\frac{1}{2}} \quad .
 \end{align}
If we apply this on the estimate \eqref{DistkleineralsquadratischeForm}, we get
 \begin{align}
&C_{con}\rho \|\tilde{A}^{-1}\|^2 \left( \lambda^2\|\nabla \tilde{A}(x)\|^2+\|\nabla \tau(x)- \tilde{A}  \|^2\right)^{\frac{1}{2}}\nonumber\\
\le&\left( \left(\alpha_\nabla\right)^{\frac{1}{2}}\left\|\tilde{A}^{-1}\right\| 
 +O(\lambda^{-1})\right)\left(2^d\|\nabla \sqrt{\tilde{\varphi}}\|_\infty^2 \rho_{2\lambda}(\chi,x)\right)^{\frac{1}{2}}\lambda^{-1}\sqrt{J}_\lambda\left(\tilde{\aff}, \chi, x\right).
\end{align}
We solve this for $J_\lambda$ and obtain for large enough $\lambda$
\begin{align}
J_\lambda\left(\tilde{\aff}, \chi, x\right)\ge & \frac{C_{con}^2 \left\|\tilde{A}^{-1}\right\|^2}{\alpha_\nabla 2^d\|\nabla \sqrt{\tilde{\varphi}}\|_\infty^2} \frac{\rho_\lambda^2}{\rho_{2\lambda}} \lambda^2\left( \lambda^2\|\nabla \tilde{A}\|^2+\|\nabla \tilde{\tau}- \tilde{A}  \|^2\right)\quad . 
\end{align}

{\bf Step two: The second derivatives}
We start with equation \eqref{implizitaffbais2}:
\begin{equation}
\partial_{\aff}^2J(\tilde{\aff}(x), \chi, x)\left(\affm,\left(\nabla_j \tilde{A}, \nabla_j \tilde{\tau}-\tilde{A}e_j\right)\right)=D_j[\affm]\quad.\nonumber 
\end{equation}
We apply the total derivative $\frac{d}{dx^k}$ on both sides , use the product rule and separate the second derivatives in $\aff$ direction from the first derivatives. We obtain
\begin{align}
&\partial_{\aff}^2J(\tilde{\aff}(x), \chi, x)\left(\affm,\left(\nabla_k\nabla_j \tilde{A}, \nabla_k\nabla_j \tilde{\tau}-\nabla_k\tilde{A}e_j\right)\right)\nonumber\\
=&\frac{d}{dx^k}D_j[\affm]
-\left(\frac{d}{dx^k}\partial_{\aff}^2J_\lambda(\tilde{\aff}(x), \chi, x)\right)\left([\affm],\left(\nabla_j \tilde{A}, \nabla_j \tilde{\tau}-\tilde{A}e_j\right)\right)\quad.
\end{align}
 We test the equation with $\affm=\left(\nabla_k\nabla_j \tilde{A}, \nabla_k\nabla_j \tilde{\tau}-\nabla_k\tilde{A}e_j\right)$, use the local convexity to estimate the left side and sum over all $j$ and $k$ to obtain
\begin{align}\label{generalsecondderivativeJ}
&C_{con}\rho \|\tilde{A}^{-1}\|^2 \left(\lambda^2\|\nabla^2 \tilde{A}\|^2+ \|\nabla^2\tilde{\tau}-\nabla \tilde{A}\|^2\right)\nonumber\\
\le&-\sum_{j,k} \left(\frac{d}{dx^k}\partial_{\aff}^2J\right)\left(\left(\nabla_k\nabla_j \tilde{A}, \nabla_k\nabla_j \tilde{\tau}-\nabla_k\tilde{A}e_j\right)     ,\left(\nabla_j \tilde{A}, \nabla_j \tilde{\tau}-\tilde{A}e_j\right)\right)\nonumber\\
&+\sum_{j,k} \frac{d}{dx^k}D_j\left(\nabla_k\nabla_j \tilde{A}, \nabla_k\nabla_j \tilde{\tau}-\nabla_k\tilde{A}e_j\right)\quad.
\end{align}
First, we calculate $\left(\frac{d}{dx^k}\partial_{\aff}^2J(\tilde{\aff}(x), \chi, x)\right)(\affm_1,\affm_2)$. We start with 
\begin{align}\label{secondderivativaffsecondderivativJ}
&C_{\varphi}\lambda^d\partial_\aff^2 J_\lambda(\aff,\chi,x)(\affm_1,\affm_2)\nonumber\\
=&\left\|\tilde{A}^{-1}\right\|^2\sum_{i} \left<M_2(x_i-x)+\mu_2,\nabla^2W (M_1(x_i-x)+\mu_1)\right>\varphi +\partial_\aff^2\left\|\tilde{A}^{-1}\right\|^2(M_1,M_2)\sum_{i}  W  \varphi  \nonumber\\
&+\partial_\aff\left\|\tilde{A}^{-1}\right\|^2[M_1]\sum_{i} \left<\nabla W ,M_2(x_i-x)+\mu_2\right>\varphi\nonumber\\
&+\partial_\aff\left\|\tilde{A}^{-1}\right\|^2[M_2]\sum_{i} \left<\nabla W,  (M_1(x_i-x)+\mu_1\right>\varphi .
\end{align}
We remember that a minimizer $\tilde{\aff}$ of $J_\lambda$ satisfies $\partial_{\aff}J(\tilde{\aff}(x), \chi,x )=0$
\begin{align} 
\frac{1}{C_{\varphi}\lambda^d}\sum_{i} \left<\nabla W , M(x_i-x)+\mu\right>\varphi
=&-\partial_A\left\|\tilde{A}^{-1}\right\|^2[M]\left\|\tilde{A}^{-1}\right\|^{-2}\frac{1}{C_{\varphi}\lambda^d}\sum_{i} W\varphi \nonumber\\
=& -\partial_A\left\|\tilde{A}^{-1}\right\|^2[M]\left\|\tilde{A}^{-1}\right\|^{-4} J_\lambda(\tilde{\aff},\chi,x) \quad .
\end{align}
Therefore, equation \eqref{secondderivativaffsecondderivativJ} turns into
\begin{align}
\partial_\aff^2 J_\lambda(\aff,\chi,x)(\affm_1,\affm_2)=&\frac{\left\|\tilde{A}^{-1}\right\|^2}{C_{\varphi}\lambda^d}\sum_{i} \left<M_2(x_i-x)+\mu_2,\nabla^2W(M_1(x_i-x)+\mu_1)\right>\varphi\nonumber\\
&+\left\|\tilde{A}^{-1}\right\|^{-2}\partial_A^2\left\|\tilde{A}^{-1}\right\|^2(M_1,M_2)J_\lambda(\tilde{\aff},\chi,x)\nonumber\\
&-2\left\|\tilde{A}^{-1}\right\|^{-4} \partial_A\left\|\tilde{A}^{-1}\right\|^2[M_1]\partial_A\left\|\tilde{A}^{-1}\right\|^2[M_2] J_\lambda(\tilde{\aff},\chi,x) \quad .
\end{align}
Next, we calculate $\frac{d}{dx}\partial_\aff^2 J_\lambda(\tilde{\aff},\chi,x)(\affm_1,\affm_2)$.
We realize that a derivative on one of the $\|\tilde{A}^{-1}\|$ terms will produce an inner derivative  $\nabla A=O(\lambda^{-2}\sqrt{J_\lambda})$.
Furthermore, $\partial_\aff\left\|\tilde{A}^{-1}\right\|^2[M]$ is $O(\lambda^{-1})\|\affm\|_\lambda$.
Hence, we get for the derivative of the second line
\begin{align}
&\left|2\frac{d}{dx}\left(\left\|\tilde{A}^{-1}\right\|^{-4} \partial_A\left\|\tilde{A}^{-1}\right\|^2(M_1)\partial_A\left\|\tilde{A}^{-1}\right\|^2(M_2) J_\lambda(\tilde{\aff},\chi,x)\right)\right|\nonumber\\
\le &O(\lambda^{-4}J_\lambda)\|\affm_1\|_\lambda\|\affm_2\|_\lambda +O(\lambda^{-2})\|\affm_1\|_\lambda\|\affm_2\|_\lambda \left|\frac{d}{dx} J_\lambda(\aff,\chi,x)\right| \quad.
\end{align}
Since it holds $\partial_\aff J_\lambda=0$, we get
\begin{align}
C_{\varphi}\lambda^d\frac{d}{dx} J_\lambda(\aff,\chi,x)=&-\left\|\tilde{A}^{-1}\right\|^2\sum_{i}\left( \nabla W\tilde{A}\varphi+2\lambda^{-1} W \sqrt{\tilde{\varphi}} \nabla \sqrt{\tilde{\varphi}}\right) \quad .
\end{align}
According to equation \eqref{gradienttauzero} the first term is zero. We apply Cauchy-Schwarz inequality on the second term, as we did in the estimate \eqref{Cauchyschwarz}. We obtain
$\left| \frac{d}{dx}J_\lambda(\tilde{\aff},\chi,x)\right|\le O(\lambda^{-1}J_\lambda)$. Therefore, it holds
\begin{align}
&\left(\frac{d}{dx^k}\partial_\aff^2 J_\lambda\right)(\affm_1,\affm_2)\nonumber\\
=&\frac{d}{dx^k}\left(\frac{\left\|\tilde{A}^{-1}\right\|^2}{C_{\varphi}\lambda^d}\sum_{i} \left<M_2(x_i-x)+\mu_2, \nabla^2W(M_1(x_i-x)+\mu_1)\right>\varphi\right)\nonumber\\
&+O(\lambda^{-3}J_\lambda)\|\affm_1\|_\lambda\|\affm_2\|_\lambda \quad .
\end{align}
The $x$ derivative can be applied on $\|\tilde{A}^{-1}\|$ producing an inner derivative 
$\nabla A=O(\lambda^{-2}\sqrt{J})$. A total  $x$-derivative of the $W$ will have an inner derivative 
\begin{equation}
\frac{d}{dx^k}(\tilde{A}(x_i-x)+\tilde{\tau})=  \left(\nabla_k \tilde{A}(x_i-x)+\nabla_k \tilde{\tau}-\tilde{A}e_k\right) = O(\lambda^{-1}\sqrt{J})\quad.
\end{equation}
Hence, we get
\begin{align}
&C_{\varphi}\lambda^d\left(\frac{d}{dx^k}\partial_\aff^2 J_\lambda(\tilde{\aff},\chi,x)\right)(\affm_1,\affm_2)\nonumber\\
=&2\lambda^{-1}\left\|\tilde{A}^{-1}\right\|^2\sum_{i} \left<M_2(x_i-x)+\mu_2,\nabla^2W(M_1(x_i-x)+\mu_1)\right>\sqrt{\tilde{\varphi}}\nabla_k\sqrt{\tilde{\varphi}} \nonumber\\
&-\left\|\tilde{A}^{-1}\right\|^2\sum_{i} \left<M_2e_k,\nabla^2W(M_1(x_i-x)+\mu_1)\right>\tilde{\varphi}\nonumber\\
&-\left\|\tilde{A}^{-1}\right\|^2 \sum_{i}\left<M_1e_k,\nabla^2W(M_2(x_i-x)+\mu_2)\right>  \tilde{\varphi}+O(\lambda^{d-3}\sqrt{J_\lambda})\|\affm_1\|_\lambda\|\affm_2\|_\lambda \quad .
\end{align}
We test with some $\affm_1^{j,k}$ and $\affm_1^{j}$, sum over $j$ and $k$ and use
Cauchy Schwarz inequality to obtain
\begin{align}
&C_{\varphi}\lambda^d\left\|\tilde{A}^{-1}\right\|^{-2}\left|\sum_{j,k}\left(\frac{d}{dx^k}\partial_\aff^2 J_\lambda(\tilde{\aff},\chi,x)\right)(\affm_1^{j,k},\affm_2^{j})\right|\nonumber\\
\le&2\lambda^{-1}\left(\sum_{i,j,k}\left(\nabla^2 W\left(M_2^j(x_i-x)+\mu_2^j\right)\right)^2 \left(\nabla_k \sqrt{\tilde{\varphi}}\right)^2\right)^{\frac{1}{2}}  \left(\sum_{i,j,k} \left(M_1^{j,k}(x_i-x)+\mu_1^{j,k}\right)^2 \tilde{\varphi}\right)^{\frac{1}{2}} 
\nonumber\\
&+ \left(\sum_{i,j,k}\left(\nabla^2W(M_2^je_k)\right)^2  \tilde{\varphi}\right)^{\frac{1}{2}} \left(\sum_{i,j,k} \left(M_1^{j,k}(x_i-x)+\mu_1^{j,k}\right)^2 \tilde{\varphi}\right)^{\frac{1}{2}} \nonumber\\
&+ \left(\sum_{i,j,k} \left( \left(M_2^j(x_i-x)+\mu_2^j\right)\nabla^2W\right)^2  \tilde{\varphi}\right)^{\frac{1}{2}} \left(\sum_{i,j,k} (M_1^{j,k}e_k)^2\tilde{\varphi}\right)^{\frac{1}{2}} \nonumber\\
&+O(\lambda^{d-3}\sqrt{J_\lambda})\|\affm_1\|_\lambda\|\affm_2\|_\lambda \quad .
\end{align}
Finally, it holds
\begin{align}\label{generaldaffdadddxfinal}
&\left|\sum_{j,k} \left(\frac{d}{dx^k}\partial_{\aff}^2J(\tilde{\aff}, \chi, x)\right)\left(\left(\nabla_k\nabla_j \tilde{A}, \nabla_k\nabla_j \tilde{\tau}-\nabla_k\tilde{A}e_j\right) ,\left(\nabla_j \tilde{A}, \nabla_j \tilde{\tau}-\tilde{A}e_j\right)\right)\right|\nonumber\\
\le&\left(16\sqrt{2^d\rho_{2\lambda}\rho_\lambda}+2\sqrt{8d}\rho_\lambda  \right)   \lambda^{-1}\left\|\tilde{A}^{-1}\right\|^2 |\nabla^2 W| \left(\lambda^2  \|\nabla^2\tilde{A}  \|^2+\|\nabla^2\tilde{\tau}-\nabla \tilde{A}\|^2 \right)^{\frac{1}{2}} \nonumber\\
&\times \left(\lambda^2  \|\nabla\tilde{A}  \|^2+\|\nabla \tilde{\tau}- \tilde{A}\|^2 \right)^{\frac{1}{2}}\quad .
\end{align}
Next, we consider 
\begin{align}
C_{\varphi}\lambda^d\left(\frac{d}{dx^k}D_j\right)[\affm]=&
\lambda^{-1}\frac{d}{dx^k}\left(\left\|\tilde{A}^{-1}\right\|^2\sum_{i} \left<\nabla W(\tilde{A}\left(x_i-x\right)+\tilde{\tau}),M(x_i-x)+\mu\right> \nabla_j\tilde{\varphi}\right)\nonumber\\
&-\lambda^{-1}\frac{d}{dx^k}\left(\partial_A\left\|\tilde{A}^{-1}\right\|^2[M]\sum_{i}  W(\tilde{A}\left(x_i-x\right)+\tilde{\tau}) \nabla_j \tilde{\varphi}\right)\quad .
\end{align}
We know from our previous calculation that $W(\tilde{A}\left(x_i-x\right)+\tilde{\tau})$ gives a $O(J_\lambda)$-contribution and $\nabla W(\tilde{A}\left(x_i-x\right)\tilde{A}+\tilde{\tau})$ gives an $O(\sqrt{J}_\lambda)$-contribution.
The inner derivative of the argument of $W$ is $O(\lambda^{-1}\sqrt{J}_\lambda)$. 
Furthermore, a derivative on  $\|\tilde{A}^{-1}\|$ will produce an inner derivative \\ \noindent $\nabla \tilde{A}=O(\lambda^{-2}\sqrt{J_\lambda})$.
Finally, $\partial_\aff\left\|\tilde{A}^{-1}\right\|^2[M]$ is $O(\lambda^{-1})\|\affm\|_\lambda$. We obtain
\begin{align}\label{generalDxderivative}
&C_{\varphi}\lambda^d\left\|\tilde{A}^{-1}\right\|^{-2}\left(\frac{d}{dx^k}D_j\right)[\affm]\nonumber\\
=&2\lambda^{-1}\sum_{i}\left<\nabla_k \tilde{A}\left(x_i-x\right)+\nabla_k \tilde{\tau}-\nabla_k \tilde{A}, \nabla^2 W \left(M(x_i-x)+\mu\right)\right> \sqrt{\tilde{\varphi}}\nabla_j\sqrt{\tilde{\varphi}}\nonumber\\
&+\lambda^{-1}\sum_{i} \left<\nabla W, M(x_i-x)+\mu\right> (\sqrt{\tilde{\varphi}} \nabla_k\nabla_j\sqrt{\tilde{\varphi}} +2\nabla_k \sqrt{\tilde{\varphi}} \nabla_j \sqrt{\tilde{\varphi}})   \nonumber\\
&-\lambda^{-1}\sum_{i} \left<\nabla W, Me_k\right> \sqrt{\tilde{\varphi}}\nabla_j \sqrt{\tilde{\varphi}}+O(\lambda^{d-3}J_\lambda)\|\affm\|_\lambda  \quad .
\end{align}
We use
$\nabla_k\sqrt{\tilde{\varphi}}(X) \nabla_j\sqrt{\tilde{\varphi}}(X)=4\sqrt{\tilde{\varphi}}(X)\nabla_j\sqrt[4]{\tilde{\varphi}}(X)\nabla_k\sqrt[4]{\tilde{\varphi}}(X)$
and denote
\begin{equation}
U_{j,k}:=\nabla_k\nabla_j \tilde{A}\left(x_i-x\right)+\nabla_k\nabla_j \tilde{\tau}-\nabla_k A e_j \quad .
\end{equation}
We test the estimate \eqref{generalDxderivative} with $\affm=\left(\nabla_k\nabla_j \tilde{A}, \nabla_k\nabla_j \tilde{\tau}-\nabla_k\tilde{A}e_j\right)$ and sum over $j$ and $k$. Applying the Cauchy Schwarz inequality we get
\begin{align}
&C_{\varphi}\lambda^d \left|\sum_{j,k} \left(\frac{d}{dx^k}D_j\right)\left(\nabla_k\nabla_j \tilde{A}, \nabla_k\nabla_j \tilde{\tau}-\nabla_k\tilde{A}e_j\right)\right|\nonumber\\
\le&  \lambda^{-1}\left\|\tilde{A}^{-1}\right\|^2 |\nabla^2W|_\infty
\left(\sum_{i,j,k}  \left(\nabla_k \tilde{A}\left(x_i-x\right)+\nabla_k \tilde{\tau}- A e_k\right)^2  (\nabla_j\sqrt{\tilde{\varphi}}   )^2    \right)^{\frac{1}{2}}  \left(\sum_{i,j,k} U_{j,k}^2 \tilde{\varphi}\right)^{\frac{1}{2}}\nonumber\\
& +\lambda^{-1}\left\|\tilde{A}^{-1}\right\|^2|\left(\sum_{i}   (\nabla W)^2  \tilde{\varphi}    \right)^{\frac{1}{2}} 
\sum_{j,k} \left(\sum_{i}  \left(\nabla_k\nabla_j \tilde{A}e_k\right)^2 (\nabla_j\sqrt{\tilde{\varphi}})^2\right)^{\frac{1}{2}}\nonumber\\
&+ \lambda^{-2}\left\|\tilde{A}^{-1}\right\|^2
\left(\sum_{i,j,k}  (\nabla W)^2  \tilde{\varphi}    \right)^{\frac{1}{2}} 
 \left(\sum_{i,j,k}  U_{j,k}^2  (\nabla_k\nabla_j\sqrt{\tilde{\varphi}})^2\right)^{\frac{1}{2}}\nonumber\\
& +8 \lambda^{-2}\left\|\tilde{A}^{-1}\right\|^2
\left(\sum_{i,j,k}    (\nabla W)^2  \tilde{\varphi}    \right)^{\frac{1}{2}}\left(\sum_{i,j,k}  U_{j,k}^2 (\nabla_k\sqrt[4]{\tilde{\varphi}} \nabla_j\sqrt[4]{\tilde{\varphi}})^2\right)^{\frac{1}{2}} \quad .
\end{align}
We can use equation \eqref{nablaWquadratkleinerW} to obtain
\begin{align}\label{Dsecondfinal}
&\left|\sum_j\sum_k \left(\frac{d}{dx^k}D_j\right)\left(\nabla_k \nabla_j \tilde{A}, \nabla_k \nabla_j \tilde{\tau}-\nabla_k\tilde{A}e_j\right)\right|\nonumber\\
\le&16  \lambda^{-1}\left\|\tilde{A}^{-1}\right\|^2|\nabla^2W|_\infty  \|\nabla\sqrt{\tilde{\varphi}}\|_\infty \sqrt{2^d \rho_\lambda \rho_{2\lambda}}\nonumber\\ 
&\times \left(\lambda^2\|\nabla \tilde{A}\|^2+\|\nabla\tilde{\tau}- \tilde{A}\|^2 \right)^{\frac{1}{2}}    
\left(\lambda^2\|\nabla^2 \tilde{A}\|^2+ \|\nabla^2\tilde{\tau}-\nabla \tilde{A}\|^2\right)^{\frac{1}{2}}\nonumber\\
&+2d\lambda^{-1}\left\|\tilde{A}^{-1}\right\|\alpha_\nabla \left(\|\nabla\sqrt{\tilde{\varphi}}\|_\infty+ \|\nabla^2 \sqrt{\tilde{\varphi}}  \|_\infty+ 4  \|\nabla\sqrt[4]{\tilde{\varphi}}\otimes\nabla\sqrt[4]{\tilde{\varphi}}\|_\infty\right)\sqrt{2^d\rho_{2\lambda}}\nonumber\\
&\times \lambda^{-1}\sqrt{J}_\lambda(\tilde{\aff},\chi,x)\left(\lambda^4\|\nabla^2 \tilde{A}\|+\|\nabla^2\tilde{\tau}-\nabla \tilde{A}\|^2\right)^{\frac{1}{2}}\quad .
\end{align}
Finally, we combine the estimates \eqref{Dsecondfinal}, \eqref{generaldaffdadddxfinal} and \eqref{generalsecondderivativeJ} to get
\begin{align}\label{generalsecondderivativeJ2}
& C_{con}\rho_\lambda \|\tilde{A}^{-1}\|^2 \left(\lambda^2\|\nabla^2 \tilde{A}\|^2+ \|\nabla^2\tilde{\tau}-\nabla \tilde{A}\|^2\right)^{\frac{1}{2}}\nonumber\\
\le&2d  \|\tilde{A}^{-1}\| \sqrt{\alpha_\nabla} \left(\|\nabla\sqrt{\tilde{\varphi}}\|_\infty^2+ \|\nabla^2\sqrt{\tilde{\varphi}}\|_\infty+ 4  \|\nabla\sqrt[4]{\tilde{\varphi}}\|^2_\infty\right)\sqrt{2^d\rho_{2\lambda}}
 \lambda^{-2}\sqrt{J}_\lambda
\nonumber\\
&+\left(16\sqrt{2^d \rho_\lambda \rho_{2\lambda}} +\sqrt{8d}\rho_\lambda \right) |\nabla^2W|_\infty  \|\nabla\sqrt{\tilde{\varphi}}\|_\infty  \lambda^{-1}\left\|\tilde{A}^{-1}\right\|^2\left(\lambda^2\|\nabla \tilde{A}\|^2+\|\nabla\tilde{\tau}- A\|^2 \right)^{\frac{1}{2}}    \quad.
\end{align}
We use the the upper bound \eqref{upperboundgradientaff} on $\lambda^2  \|\nabla\tilde{A}  \|^2+\|\nabla \tau- A\|^2$ to finally arrive at
\begin{align}
&C_{con}\rho_\lambda \|\tilde{A}^{-1}\|^2\left(\lambda^2\|\nabla^2 \tilde{A}\|^2+ \|\nabla^2\tilde{\tau}-\nabla \tilde{A}\|^2\right)^{\frac{1}{2}}\nonumber\\
\le& \lambda^{-2} \|\tilde{A}^{-1}\|
\alpha_\nabla^{\frac{1}{2}} \sqrt{J}_\lambda(\tilde{\aff},\chi,x)00 \left(\left(\|\nabla \sqrt{\tilde{\varphi}}\|_\infty^2+ \|\nabla^2\sqrt{\tilde{\varphi}}\|_\infty+ 2 \infty \|\nabla\sqrt[4]{\tilde{\varphi}}\|^2\right)d\sqrt{2^d\rho_{2\lambda}}\right.\nonumber\\
&\quad +\left. C_{con}^{-1}\left(2^d\|\nabla \sqrt{\tilde{\varphi}}\|_\infty^2\right)^{\frac{1}{2}}\left(16 2^{\frac{d}{2}} \rho_\lambda^{-\frac{1}{2}} \rho_{2\lambda} +\sqrt{8d}\sqrt{\rho_{2\lambda}}\right)\right)\quad .
\end{align}
\end{proof}

\subsection{Proof of Theorem \ref{TheoremLowerboundoftheensitywiththelocalminimizersofJ}}
\begin{proof}
{\bf Step 1: Following one minimizer}
According to lemma \ref{lowenergyregular} for a $x$ point of  energy density $\hat{h}_\lambda(\chi,x)\le \epsilon\le \frac{1}{4}  \min \left\{C_1^{El}\det(E)^2,C_2^{El}|E|^2,\vartheta \det E \right\}$
there exists $\hat{\aff}=(\hat{A},\hat{\tau})$ satisfying $h_\lambda(\hat{\aff},\chi,x)=\hat{h}_\lambda(\chi,x) $.
Moreover $(x,\hat{\aff})$ is $(\epsilon_\rho, \epsilon_J, C_A)$-regular,  where
\begin{equation}
C_A= \frac{3^{d-1}|E|^{d-1}}{2^{d-2}\det E}    \quad, \quad \epsilon_\rho=   2 \frac{\epsilon}{\det E}  							\quad, \quad 
\epsilon_J=     \frac{4 \epsilon }{\det E}     					\quad .
\end{equation}
According to Lemma \ref{BasStadev} for every reparametrisation $\affb(x)=(B(x),t(x))\in Gl_d(\Z^d)\times \Z^d$ with $|(B\hat{A})^{-1}|< 2C_A$ it holds
\begin{equation}
 J_{\lambda}\left(\affb\hat{\aff}(x), \chi, x\right)\le C_1^W\left(C^W_0\right)^{-1}\left\|B\hat{A}(x)\right\|^2\left\|(B\hat{A}(x))^{-1}\right\|^2 J_{\lambda}\left(\hat{\aff}(x), \chi ,x\right).
\end{equation}
We obtain with the help of $|\hat{A}|\le|\hat{A}^{-1}|^{d-1} \det \hat{A}$ and $\det \hat{A} \le \frac{3}{2}\det E$ from Lemma \ref{lowenergyregular} 
\begin{align}\label{Lambdaestimate2}
 J_{\lambda}\left(\affb\hat{\aff}(x), \chi, x\right)\le& \frac{C_1^W}{C^W_0} \left(\det B\hat{A}(x)\right)^2 \left\|(B\hat{A}(x))^{-1}\right\|^{2d} J_\lambda\left(\hat{\aff}(x),\chi, x\right)
 \le C_{rep} J_\lambda\left(\hat{\aff},\chi, x\right) ,
\end{align}
where $C_{rep}:=9 \left(C^W_0\right)^{-1} 4^{d-1} C_A^{2d} \det E^2$.
Since the density $\rho_\lambda$ does not depend on $\aff$ and it holds $\det A=\det BA$, the position $(x,\affb\hat{\aff}(x))$ is $(3C_A, \epsilon_\rho,  C_{rep} \epsilon_J )$-regular. For large enough $\lambda$ and sufficiently small $\epsilon$ the conditions of Lemma \ref{localaffgradient} are fulfilled, and there exists a unique local minimizer $\tilde{\aff}_B$ in a neighborhood of $\affb\hat{\aff}$.
Furthermore, we get the estimate \eqref{Jlocalminimizerabstand} for the distance between $\affb\hat{\aff}$ and $\tilde{\aff}_B$.
Due to the estimate \eqref{Lambdaestimate2} we have for sufficiently small $\hat{\epsilon}$
\begin{align}\label{energyabstandA}
\left|\affb\hat{\aff}(x)-\tilde{\aff}_B(x)\right\|_\lambda\le& \left(  \frac{1}{2} C_{Con} \|A_0^{-1}\|^2\rho_\lambda \right)^{-1/2} \sqrt{J_\lambda}(\affb\hat{\aff}(x),\chi ,x)\nonumber\\
\le&\left(  \frac{1}{2} C_{Con} \|A_0^{-1}\|^2\rho_\lambda \right)^{-1/2} \hat{\epsilon} \le \frac{1}{8} \delta_\aff \quad .
\end{align}
Additionally we have the estimate \eqref{upperboundgradientaff} from Lemma \ref{Upperboundgradients} for the gradients in this branch. Hence, we get
\begin{align}
J_\lambda\left(\hat{\aff}, \chi, x_0\right)\ge&C_{rep}^{-1}J_\lambda\left(\tilde{\aff}_B, \chi, x_0\right)\nonumber\\
\ge&C_{rep}^{-1}\frac{C_{con}^2 \left\|\tilde{A}_B^{-1}\right\|^2}{\alpha_\nabla 2^d\|\nabla \sqrt{\tilde{\varphi}} \|_\infty^2} \frac{\rho_\lambda^2}{\rho_{2\lambda}} 
\lambda^2\left( \lambda^2\|\nabla \tilde{A}_B(x)\|^2+\|\nabla \tilde{\tau}(x)_B- \tilde{A}_B  \|^2\right)\quad .
\end{align}
We consider a second point $y=y(s) \in B_{1,5\lambda}(x)$ with $\hat{\lambda}(\chi,y)\le\hat{\epsilon}$ and $\int \frac{dy}{d\tilde{s}}d\tilde{s}\le \delta x$ sufficiently small and obtain
\begin{align}\label{extrapolationcontinuierlich}
\left|\tilde{A}_B(y)-\tilde{A}_B(x)\right|\le&O(\lambda^{-2}\delta x)\sqrt{J}_\lambda\left(\hat\aff, \chi, x_0\right)
\le O(\lambda^{-2}\delta x\sqrt{\hat{\epsilon}}) \le \frac{1}{8} \delta_\aff \quad, \nonumber\\
\left|\tilde{\tau}_B(y)-\tilde{\tau}_B(x)-\tilde{A}_B(x)(y-x)\right| \le &O(\lambda^{-1}|x-y|)\sqrt{J}_\lambda\left(\hat\aff, \chi, x_0\right)\le  O(\lambda^{-1}\delta x\sqrt{\hat{\epsilon}})\le \frac{1}{8}\delta_\aff \quad .
\end{align}
For $\hat{\epsilon}(x)$ and $\hat{\epsilon}(y)$ small enough the points $x$ and $y$ are regular according due to Lemma \ref{lowenergyregular}, and we can use Lemma \ref{Theoremjumping} to obtain  $\affb(x,y)=(B(x,y),\tau(x,y))\in Gl_d(\Z)\times \Z^d$ such that
\begin{align}
\|Id- \hat{A}(x)^{-1}B(x,y)\hat{A}(y)\|<&  \frac{c^{A}_J}{\sqrt{\det \hat{A}(y)}}  2^d\lambda^{-1} \min\left\{\sqrt{J}_\lambda(\hat{\aff}(x),\chi, x),\sqrt{J}_\lambda(\hat{\aff}(y),\chi, y)\right\}\nonumber\\
\le& O\left(\lambda^{-1}\sqrt{\hat{\epsilon}}\right)\quad ,
\end{align}
and
\begin{align}
&\left|B(x,y)\hat{\tau}(y)+t(x,y) -\hat{\tau}(x)-\frac{B(x,y)\hat{A}(y)+\hat{A}(x)}{2}\left(y-x\right)\right|\nonumber\\
< & \frac{c^{\tau}_J}{\sqrt{\det \hat{A}(y)}}  2^d\min\left\{\sqrt{J}_\lambda(\hat{\aff}(x),\chi, x),\sqrt{J}_\lambda(\hat{\aff}(y),\chi, y)\right\}\le 
O\left(\sqrt{\hat{\epsilon}}\right)
\end{align}
For small enough $\hat{\epsilon}$ and large enough $\lambda$ we can control the change of $\hat{\tau}$ and $\hat{A}$, because we restricted $B$ to a compact set.
\begin{align}\label{extrapolationdiscrete}
\lambda\left|B(x)\left(\hat{A}(x)-B(x,y)\hat{A}(y)\right)\right|\le&\frac{1}{8}  \delta_A \quad ,\nonumber\\
\left|B(x)\left(B(x,y)\hat{\tau}(y)+t(x,y) -\hat{\tau}(x)-\frac{B(x,y)\hat{A}(y)+\hat{A}(y)}{2}\left(y-x\right)\right)\right|\le&\frac{1}{8} \delta_\aff \quad .
\end{align}
We introduce the notation $\affb(y):=B(x)B(x,y)$. By comparing the estimates \eqref{energyabstandA}, \eqref{extrapolationcontinuierlich} and \eqref{extrapolationdiscrete} we obtain
\begin{align}
&\left|B(y)\hat{A}(y)-\tilde{A}_B(y)\right|\nonumber\\&\le\left|B(y)\hat{A}(y)-B(x)\hat{A}(x)\right|+\left|B(x)\hat{A}(x)-\tilde{A}_B(x)\right|+\left|\tilde{A}_B(x)-\tilde{A}_B(y)\right|\le \frac{3}{8\lambda}\delta_\aff.
\end{align}
For $\tau$ we estimate
\begin{align}
&|B(x)\left(B(x,y)\hat{\tau}(y)+t(x,y)-\tilde{\tau}_B(y)\right)|\nonumber\\
\le&\left|B(x)\left(B(x,y)\hat{\tau}(y)+t(x,y) -\hat{\tau}(x)-\frac{B(x,y)\hat{A}(y)+\hat{A}(x)}{2}\left(y-x\right)\right)\right|\nonumber\\
&+\left|B\hat{\tau}(x)+t-\tilde{\tau}_B(x)\right|+\left|\tilde{\tau}_B(x)+\tilde{A}_B(x)(y-x)-\tilde{\tau}_B(y)\right|\nonumber\\
&+\left|\tilde{A}_B(x)-B(x)\frac{B(x,y)\hat{A}(y)+\hat{A}(x)}{2}\left(y-x\right)\right|\nonumber\\
\le&\frac{3}{8}\delta_\aff+\frac{3}{4}\lambda|\tilde{A}_B(x)-B(y)\hat{A}(y)|+\frac{3}{4}\lambda |\tilde{A}_B(x)-B(x)\hat{A}(x)|
\le\frac{21}{32} \delta_\aff  .
\end{align}
We summarize 
\begin{equation}\label{distanceminimizer}
\left\|\affb(y)\hat{\aff}(y)-\tilde{\aff}_B(y)\right\|_\lambda \le \delta_\aff \quad .
\end{equation}
Since  $\affb(y)\hat{\aff}(y)$ fulfills the same conditions for $y$ as $\affb(x)\aff(x)$ for $x$ we can apply Theorem \ref{localaffgradient}. Hence, there is one unique local minimizer satisfying
\begin{equation}
\left\|\affb(y)\hat{\aff}(y)-\tilde{\aff}\right\|_\lambda \le \delta_\aff \quad.
\end{equation}
Therefore, $\tilde{\aff}_B(y)$ has to be this minimizer because of the estimate \eqref{distanceminimizer}. 

{\bf Step 2: The lower bound for the energy density:} 
Due to estimate \eqref{Jlocalminimizerabstand} from Lemma \ref{localaffgradient} we get for $\tilde{\aff}_B(y)$
\begin{align}
J_\lambda\left(\hat{\aff}(y), \chi, y\right)\ge&C_{rep}^{-1} J_\lambda\left(\affb(y)\hat{\aff}(y), \chi, y\right)\nonumber\\
\ge &C_{rep}^{-1}J_\lambda(\tilde{\aff}_B,\chi,x)+\frac{1}{2} C_{Con}C_{rep}^{-1} \left\|\left(B\hat{A}\right)^{-1}\right\|^2 \rho_\lambda \left\|\affb\hat{\aff}-\tilde{\aff}_B\right\|_\lambda^2 \quad .
\end{align}
Applying Lemma \ref{Upperboundgradients}, we get
\begin{align}\label{estimategraident2}
J_\lambda\left(\hat{\aff}(y), \chi, y\right)\ge&
\frac{1}{2} C_{Con}C_{rep}^{-1} \left|\left(B\hat{A}\right)^{-1}\right\|^2\rho_\lambda \left\|\affb(y)\hat{\aff}-\tilde{\aff}_B\right\|_\lambda^2\nonumber\\
&+\tilde{C}_{\nabla }\left(\frac{\rho_{2\lambda}}{\rho_\lambda}\right) \left\|\tilde{A}_B^{-1}\right\|^2\rho_\lambda
\left( \lambda^2\|\nabla \tilde{\tau}_B- \tilde{A}_B  \|^2 + \lambda^6\|\nabla^2 \tilde{A}_B\|^2\right)\nonumber\\
&+\tilde{C}_{\nabla }\left(\frac{\rho_{2\lambda}}{\rho_\lambda}\right) \left\|\tilde{A}_B^{-1}\right\|^2\rho_\lambda \lambda^4\left( \|\nabla \tilde{A}_B\|^2+\|\nabla^2\tilde{\tau}_B-\nabla \tilde{A}_B\|^2\right) \quad ,
\end{align}
We apply the estimates \eqref{estimategraident2} to get a lower bound for the density
\begin{align}
\hat{h}_\lambda(\chi,y)=&J_\lambda\left(\hat{\aff}(y),\chi,y\right)+\nu(\hat{A}(y),\chi, y)+F(\hat{A}(y))\nonumber\\
\ge& F(\hat{A})+\frac{1}{2} C_{Con}C_{rep}^{-1} \left\|\left(B\hat{A}\right)^{-1}\right\|^2 \rho_\lambda \left\|\affb\hat{\aff}-\tilde{\aff}_B\right\|_\lambda^2\nonumber\\
&+\tilde{C}_{\nabla }\left(\frac{\rho_{2\lambda}}{\rho_\lambda}\right) \left\|\tilde{A}_B^{-1}\right\|^2\rho_\lambda\left( \lambda^2\|\nabla \tilde{\tau}_B- \tilde{A}_B  \|^2 + \lambda^6\|\nabla^2 \tilde{A}_B\|^2\right)\nonumber\\
&+\tilde{C}_{\nabla }\left(\frac{\rho_{2\lambda}}{\rho_\lambda}\right) \left\|\tilde{A}_B^{-1}\right\|^2\rho_\lambda
\left(\lambda^4 \|\nabla \tilde{A}_B\|^2+\lambda^4\|\nabla^2\tilde{\tau}_B-\nabla \tilde{A}_B\|^2\right).
\end{align}
Since we calculate a lower bound, we can skip the $\nabla^2 \tilde{A}_B$ term. We also estimate
\begin{equation}
\left\|\affb\hat{\aff}-\tilde{\aff}_B\right\|_\lambda^2\ge\lambda^2\left\|B\hat{A}-\tilde{A}_B\right\|^2 \quad .
\end{equation}
Due to $2(a^2+b^2)\ge (a+b)^2$ we summarize
\begin{equation}
\lambda^4 \|\nabla \tilde{A}_B\|^2+\lambda^4\|\nabla^2\tilde{\tau}_B-\nabla \tilde{A}_B\|^2\ge\frac{1}{2}\lambda^4\|\nabla^2\tilde{\tau}_B\|^2 \quad .
\end{equation}
Due to the estimate \ref{estimategraident2} the difference between $\left\|\tilde{A}_B^{-1}\right\|^2$ and $\left\|\nabla\tilde{\tau}_B^{-2}\right\|^2$ is $O\left(\lambda^{-1}\sqrt{\hat{\epsilon}}\right)$. We estimate for small $\hat{\epsilon}$
\begin{align}
\rho_\lambda=&\det \hat{A}+O(\sqrt{\hat{\epsilon}})=\det \tilde{A}_B + O(\sqrt{\hat{\epsilon}})=\det \nabla \tilde{\tau}_B+ O(\sqrt{\hat{\epsilon}})\quad .
\end{align}
Hence, we get for small enough $\hat{\epsilon}$ and large enough $\lambda$.
\begin{align}
\hat{h}_\lambda(\chi,y) \ge& F(\hat{A})+\frac{1}{2}\tilde{C}_{\nabla }\left(\frac{\rho_{2\lambda}}{\rho_\lambda}\right)  \left\|\tilde{A}_B^{-1}\right\|^2  \det(\nabla \tilde{\tau}_B) \lambda^2\|\nabla \tilde{\tau}_B- \tilde{A}_B  \|^2\nonumber\\
&+\frac{1}{3} C_{Con}C_{rep}^{-1} \left\|\left(B\hat{A}\right)^{-1}\right\|^2    \left\|\affb\hat{\aff}-\tilde{\aff}_B\right\|_\lambda^2\det(\nabla \tilde{\tau}_B)\nonumber\\
&+\frac{1}{2}\tilde{C}_{\nabla }\left(\frac{\rho_{2\lambda}}{\rho_\lambda}\right)\left\|\nabla\tilde{\tau}_B^{-1}(y)\right\|^2\lambda^4\|\nabla^2\tilde{\tau}_B\|^2\det(\nabla \tilde{\tau}_B) \quad .
\end{align}
We summarize all but the $\|\nabla^2 \tilde{\tau}_B\|^2$ term to $U(\tilde{\tau}_B, \tilde{A}_B, B(y),\hat{A})$
\begin{align}
\hat{h}_\lambda(\chi,y)\ge& \frac{1}{2}\tilde{C}_{\nabla }\left(\frac{\rho_{2\lambda}}{\rho_\lambda}\right)\left\|\nabla\tilde{\tau}_B^{-1}\right\|^2\lambda^4\|\nabla^2\tilde{\tau}_B\|^2\det(\nabla \tilde{\tau}_B)\nonumber\\
&+U(\tilde{\tau}_B, \tilde{A}_B, B, \hat{A}) \quad .
\end{align}
Finally, we use
\begin{equation}
U(\tilde{\tau}_B, \tilde{A}_B, B(y), \hat{A})\ge \inf\left\{U(\tilde{\tau}_B, A_1, B, A_2)| A_1,A_2\in Gl_d(\R), B\in Gl_d(\Z)   \right\}\quad .
\end{equation}
\end{proof}

\appendix

\section{Basic calculations}

\begin{lemma}\label{condistapp}
For all $C_A$ exists $\hat{\lambda}$ such that for all $\lambda>\hat{\lambda}$ all $A_R\in Gl_d(\R)$, $\tau_R\in \R^d$ 
and $x\in B_{2\lambda}(\Omega)$ and $\psi\in C^{infty}(\R^d)$ with $\psi(y)=0$ for $|y|>1$ it holds
\begin{equation}
\lambda^{d}\int_{\R^d}\psi\left(y\right)dy \det A_R= \sum_{x_i\in \chi_\aff} \psi\left(\frac{x_i-x}{\lambda}\right)+ O\left(\lambda^{-2}|\nabla^2 \psi|_\infty\right)
\end{equation}
In particular it holds
\begin{equation}
\rho_\lambda\left(\chi_{\aff_R},0\right)=\det A_R+O(\lambda^{-2})
\end{equation}
\end{lemma}

\begin{proof}
OBDA we can restrict ourselves to $x=0$. We denote $Q_i:=x_i+[-1/2,1/2)^d$. We calculate
\begin{align}
\lambda^{d}\int_{\R^d}\psi\left(y\right)dy=&\int_{\R^d}\psi\left(\lambda^{-1}y\right)dy
=\sum{x_i \in \chi_{\aff_R}}\int_{Q_i}\psi\left(\lambda^{-1}y\right)dy\nonumber\\
=&\sum_{x_i \in \chi_{\aff_R}}\int_{Q_i}\psi\left(\lambda^{-1}x_i\right)+\nabla \psi\left(\lambda^{-1}x_i\right) \left[\frac{y-x_i}{\lambda}\right]+\frac{1}{2}\nabla^2 \psi\left(\lambda^{-1}x_i\right)\left[\frac{y-x_i}{\lambda}\right]dy\nonumber\\
=&\det A_R^{-1}\lambda^d \sum_{x_i\in \chi_\aff} +O\left(\lambda^{-2}|\nabla^2 \psi|_\infty\right)
\end{align}
\end{proof}

\begin{lemma}\label{BasStadev}
For all $\aff=(A,\tau)\in Gl_d(\R)\times R^d$, all positions $x$ and configurations $\chi$ it holds
\begin{align}\label{GLStandartdeviation}
 J_{\lambda}\left(\aff, \chi,x\right) \ge&\frac{C_0^W}{C_{\varphi}\lambda^d} \sum_{i} \mathrm{dist}^2(x_i, \chi_{\aff}+x ) \varphi\left(\lambda^{-1}\left|x_i-x\right|\right)\quad ,\nonumber\\
 J_{\lambda}\left(\aff, \chi,x\right) \le& \frac{C_1^W\left\|A\right\|^2\left\|A^{-1}\right\|^2}{C_{\varphi}\lambda^d} \sum_{i} \mathrm{dist}^2(x_i, \chi_{\aff}+x ) \varphi\left(\lambda^{-1}\left|x_i-x\right|\right) \quad.
\end{align}
In particular, for $\affb=(B.t)\in Gl_d(\Z) \times \Z^d$ it holds
\begin{equation}
 J_{\lambda}\left(\aff, \chi, x\right) \le  \frac{C_1^W\left\|A\right\|^2\left\|A^{-1}\right\|^2}{C^W_0} J_{\lambda}\left(\affb\aff, \chi ,x\right)\quad .
\end{equation}
\end{lemma}

\begin{proof}
On the one hand we have
\begin{align}
C_{\varphi}\lambda^d J_{\lambda}\left(\aff, \chi,x\right)
=&\left\|A^{-1}\right\|^2\sum_{i} W(A\left(x_i-x\right)+\tau)\varphi\left(\lambda^{-1}\left|x_i-x\right|\right)\nonumber\\
\leq &C_1^W\left\|A^{-1}\right\|^2\sum_{i} \mathrm{dist}^2 (A\left(x_i-x\right)+\tau, \Z^d)\varphi\left(\lambda^{-1}\left|x_i-x\right|\right)     \nonumber\\
\leq &C_1^W\left\|A^{-1}\right\|^2\left\|A\right\|^2  \sum_{i}  \mathrm{dist}^2 (x_i, A^{-1}(\Z^d-\tau)+x)\varphi\left(\lambda^{-1}\left|x_i-x\right|\right).     
\end{align}
On the other hand we have
\begin{align}
C_{\varphi}\lambda^d J_{\lambda}\left(\aff , \chi,x\right)
=&\left\|A^{-1}\right\|^2\sum_{i} W(A\left(x_i-x\right)+\tau)\varphi\left(\lambda^{-1}\left|x_i-x\right|\right)\nonumber\\
\ge &C_0^W\left\|A^{-1}\right\|^2\sum_{i} \mathrm{dist}^2 (A\left(x_i-x\right)+\tau, \Z^d)\varphi\left(\lambda^{-1}\left|x_i-x\right|\right)     \nonumber\\
\ge &C_0^W \sum_{i}  \mathrm{dist}^2 (x_i, A^{-1}(\Z^d-\tau)+x)\varphi\left(\lambda^{-1}\left|x_i-x\right|\right). 
\end{align}
\end{proof}

\begin{lemma}\label{Dichteregular}
If  $x\in B_{2\lambda}(\Omega)$ and $\aff\in Gl_d(\R)\times \R^d$, we have
\begin{equation}\label{eqDichteregular}
\rho_{\aff,\beta}^{irr}(x)\le \frac{1}{C_0^W \beta^2 } J_\lambda(\aff,\chi,x) \quad, \quad \rho_{\aff, \beta}^{reg}(x)\ge\rho_\lambda(\chi,x)-\frac{1}{C_0^W \beta^2 } J_\lambda(\aff,\chi,x) \quad. 
\end{equation}
\end{lemma}
\begin{proof}
We use equation \eqref{GLStandartdeviation} to get
\begin{align}
 J_{\lambda}\left(\aff, \chi_{\aff_R},x\right) \ge& \frac{C_0^W}{C_{\varphi}\lambda^d} \sum_{i\in\chi} \mathrm{dist}^2( x_i, \chi_{\aff}+x ) \varphi\left(\lambda^{-1}\left|x_i-x\right|\right)\nonumber\\
\ge&\frac{C_0^W}{C_{\varphi}\lambda^d} \sum_{x_i\in\chi_{\aff, \beta, x}^{irr}} \mathrm{dist}^2( x_i, \chi_{\aff}+x ) \varphi\left(\lambda^{-1}\left|x_i-x\right|\right)\nonumber\\
\ge&\frac{C_0^W}{C_{\varphi}\lambda^d} \sum_{x_i\in \chi^{irr}_{\aff, \beta, x}} \beta^2 \varphi\left(\lambda^{-1}\left|x_i-x\right|\right)\nonumber\\
\ge&C_0^W \beta^2 \rho_{\aff , \beta}^{irr}\quad .
\end{align}
Because it holds $\rho_\lambda=\rho_{\aff, \beta}^{irr}+\rho_{\aff, \beta}^{reg}$, we obtain equation \eqref{eqDichteregular}
\end{proof}

\section{Estimate on the change of $\aff$ in an sequence of regular points.}

\begin{lemma}\label{Ajumpingchain}
For all $C_A>0$ there exists $\hat{\lambda}$ $\epsilon_\rho$ $\epsilon_J$ such that for all $\lambda>\hat{\lambda}$ the following holds:
If there is a exists $(\aff_j,y_j)\in Gl_d(\R)\times \R^d\times \Omega $, $x_j$  is $(\epsilon_\rho, \epsilon_J, C_A)$-regular with $\aff_j$ 
and  $\affb_{j-1,j}=(B_{j-1,j}, t_{j-1,j})$ denotes the associated reparametrisation sequence given by Theorem \ref{addingatoms} for $j=0...N$ ,
then it holds 
\begin{align}
\left|1-A_0^{-1}B_{0,N}A_N  \right|
\leq& \frac{c^{A}_J}{\lambda}\sum_{j=1}^N\hat{b}_{j-1,j} \exp \left(\frac{c^{A}_J}{\lambda}\sum_{j=1}^N\hat{b}_{j-1,j}\right)\quad,
\end{align}
and
\begin{align}
&\left|\sum_{k=1}^N B_{0,k-1}t_k+B_{0,N}\tau_N-\tau_0+\frac{B_{0,N}A_N+A_0}{2}\left(y_N-y_0\right)\right|\nonumber\\
\leq&\left(C_AC_{|A|}c^{\tau}_J +\frac{c_J^A}{\lambda}\sum_{j=1}^N \left|y_{j+1}-y_j\right|\right)
\left|A_0\right|\sum_{j=1}^{N}\hat{b}_{j-1,j} \exp \left(\frac{C_A}{\lambda} \sum_{k=1}^N\hat{b}_{k-1,k}\right)\quad ,
\end{align}
where
\begin{align}
J_j=&J_\lambda( \aff_j,\chi, y_j)\quad ,\nonumber\\
\hat{b}_{j-1,j}=&\left(\frac{2\lambda}{2\lambda-|y_j-y_{j-1}|}\right)^{d/2}\left(\det A_j\right)^{-1/2} \max\left\{\sqrt{J_j},\sqrt{J_{j-1}}\right\} \quad ,\nonumber\\
B_{k_1,k_2}=&\prod_{j=k_1+1}^{k_2}B_{j-1,j}\quad.
\end{align}
\end{lemma}

\begin{proof}
{\bf the $B$-Product}
We use the notation
\begin{align}
a_{j-1,j}:=&A_{j-1}^{-1}B_{j-1,j}A_j \quad ,\nonumber\\
b_{j-1,j}:=&1-A_{j-1}^{-1}B_{j-1,j}A_j=1-a_{j-1,j} \quad .
\end{align}
Due to Theorem \ref{Theoremjumping} and we have for every $j=1...N$
\begin{align}
|b_{j-1,j}|\le&\|b_{j-1,j}\|<\hat{b}_{j-1,j}\frac{c^{A}_J}{\lambda}\quad .
\end{align}
Using this upper bound for $|b_{j-1,j}|$ we derive an upper bound for general products of $a_{j-1,j}$
\begin{align}\label{aproduct}
\left|\prod_{j=k_1+1}^{k_2}a_{j-1,j} \right|=&\left|\prod_{j=k_1+1}^{k_2}(1-b_{j-1,j})\right|\le \prod_{j=k_1+1}^{k_2}\left(1+\left|b_{j-1,j}\right|\right)\nonumber\\
\le&\prod_{j=k_1+1}^{k_2}\exp\left(\left|b_{j-1,j}\right|\right)\le \exp\left(\sum_{j=k_1+1}^{k_2}\left|b_{j-1,j}\right|\right)\quad .
\end{align}
Furthermore, we get
\begin{equation}
\prod_{j=k_1+1}^{k_2}a_{j-1,j}=\prod_{k_1+1}^{k_2} \left(A_{j-1}^{-1}B_{j-1,j}A_j\right)=A_{k_1}^{-1} B_{k_1,k_2}A_{k_2}\quad.
\end{equation}
We derive a bound on $1-\prod a_{j-1,j}$
\begin{align}\label{Glaschranke}
\left|1-\prod_{j=1}^Na_{j-1,j}\right|\le&\left|1-a_{0,1}+\sum_{k=2}^N\prod_{j=1}^{k-1}a_{j-1,j}(1-a_{k-1,k})\right|\nonumber\\
\le&\left|1-a_{0,1}\right|+\sum_{k=2}^N\left|1-a_{k-1,k}\right| \left|\prod_{j=1}^{k-1}a_{j-1,j}\right|    \nonumber\\
\le&\left|b_1\right|+\sum_{k=2}^N\left|(b_{k-1,k})\right|\exp\left(\sum_{j=1}^{k-1}\left|b_{j-1,j}\right|\right)\nonumber\\
\le&\sum_{k=1}^N\left|(b_{k-1,k})\right| \exp\left(\sum_{j=1}^{N}\left|b_{j-1,j}\right|\right)\nonumber\\
\leq&\sum_{j=1}^N\frac{C^{A}_J}{\lambda}\hat{b}_{j-1,j} \exp\left(\frac{C^{A}_J}{\lambda}\sum_{k=1}^{N}\hat{b}_{k-1,k}\right)\quad .
\end{align}

{\bf The $\tau$ product }
We denote
\begin{align}
\delta \tau:=& \sum_{j=1}^N   B_{0,j-1} t_{j-1,j} +   B_{0,N}\tau_N -\tau_0+\frac{B_{0,N}A_N+A_0}{2}\left(y_N-y_0\right)\quad ,\nonumber\\
\delta \tau_j:=& t_{j-1,j}+ B_{j}\tau_{j} -\tau_{j-1}+\frac{B_{j-1,j}A_{j}+A_{j-1}}{2}\left(y_{j+1}-y_j\right)\quad ,\nonumber\\
\end{align}
Since $x_j$ is $(2^{-3-2d}, \epsilon_J, C_A)$-regular with $\aff_j$ 
we can use Theorem \ref{Theoremjumping} and get for  every $j=1...N$
\begin{align}\label{aiestimate}
\|1-a_{j-1,j}\|<& \frac{c^{A}_J}{\lambda}\hat{b}_{j-1,j} \quad ,  \nonumber\\
\left|\delta \tau_j\right|<& c^{\tau}_J  \|A_{j-1}\| \hat{b}_{j-1,j} \quad.
\end{align}
Hence, we have bounds for $\delta \tau_j$ and want a bound for $\delta \tau$
\begin{align}\label{GLdeltatauchain1}
\left|\delta \tau\right|=&\left|\sum_{j=1}^{N}B_{0,j-1} t_{j-1,j} + B_{0,j}\tau_{j} -B_{0,j-1}\tau_{j-1}+
\frac{B_{0,N}A_N+A_0}{2}\left(y_{j}-y_{j-1}\right)\right|\nonumber\\
=&\left|\sum_{j=1}^{N} B_{0,j-1}\delta \tau_j+\frac{1}{2}\left(B_{0,N}A_N-B_{0,j}A_{j}+A_0-B_{0,j-1} A_{j-1}     \right)\left(y_{j}-y_{j-1}\right)\right|\nonumber\\
\le&c^{\tau}_J \sum_{j=1}^{N} \left|B_{0,j-1}\right| \|A_{j-1}\| \hat{b}_{j-1,j}\nonumber\\
&+\frac{1}{2}\sum_{j=1}^{N}\left(\left|B_{0,N}A_N-B_{0,j}A_j\right|+\left|A_0-B_{0,j-1} A_{j-1}\right|\right)    \left|y_{j}-y_{j-1}\right|
\quad.
\end{align}
Due to the inequality \eqref{aproduct} we can estimate
\begin{align}\label{Biestimate}
\left|B_{0,j-1}\right|\le& \left|A_0\right|\left|A_0^{-1}B_{0,j-1}A_{j-1}\right|\left|A_{j-1}^{-1}\right|\le C_A\left|A_0\right|\exp \left(\frac{c^{A}_J}{\lambda}\sum_{k=1}^{N}\hat{b}_{k-1,k}\right)\quad.
\end{align}
We  calculate
\begin{align}\label{difference1}
\left|A_0-B_{0,n-1} A_{n-1}\right| \le &\sum_{j=1}^{n-1}\left|B_{0,j}A_j-B_{0,j-1}A_{j-1} \right| \nonumber\\
 \le &\sum_{j=1}^{n-1}\left|B_{0,j-1}A_{j-1}\right| \left|A_{j-1}^{-1}B_{j-1,j}A_j-id\right| \nonumber\\
\le & \sum_{j=1}^{n-1}\left|A_0\right|\left|A_0^{-1}B_{0,j-1}A_{j-1}\right| \left|A_{j-1}^{-1}B_{j-1,j}A_j-id\right|\nonumber\\
\le & \sum_{j=1}^{n-1}C_{|A|}\left|\prod_{k=1}^{j-1}a_{k-1,k}\right| \left|a_{j-1,j}-id\right|\nonumber\\
\le &C_{|A|}\sum_{j=1}^{n-1}\frac{c^{A}_J}{\lambda}\hat{b}_{j,j-1} \exp \left( \frac{c^{A}_J}{\lambda}\sum_{k=1}^N\hat{b}_{k-1,k}\right)\quad .
\end{align}
 We estimate  $\left|B_{0,N}A_N-B_{0,n}A_n\right|$ in the same way and obtain
\begin{align}\label{difference2}
\left|B_{0,N}A_N-B_{0,n}A_n\right| \le &C_{|A|}\sum_{j=n+1}^{N}\frac{c^{A}_J}{\lambda}\hat{b}_{j-1,j} \exp \left(\frac{c^{A}_J}{\lambda}\sum_{k=1}^N\hat{b}_{k-1,k}\right) \quad. 
\end{align}
A combination of the estimates \eqref{difference1} and \eqref{difference2} leads to
\begin{equation}\label{deltaAestimate}
\left|A_0-B_{0,n-1} A_{n-1}\right|+\left|B_{0,N}A_N-B_{0,n}A_n\right|\le C_{|A|}\sum_{j=1}^{N}\frac{c^{A}_J}{\lambda}\hat{b}_{j-1,j} \exp \left(\frac{c^{A}_J}{\lambda}\sum_{k=1}^N\hat{b}_{k-1,k}\right) \quad.
\end{equation}
Using the estimates \eqref{aiestimate}, \eqref{Biestimate} and \eqref{deltaAestimate} results in
\begin{align}
\left|\delta \tau\right|\le&\sum_{j=1}^N C_{|A|}c^{\tau}_JC_A\left|A_0\right|\hat{b}_{j-1,j}\exp \left(\frac{c^{A}_J}{\lambda} \sum_{k=1}^{N}\hat{b}_{k-1,k}\right) \nonumber\\
&+ C_{|A|}\sum_{j=1}^{N}\frac{C_J^A}{\lambda}\hat{b}_{j-1,j} \exp \left(\frac{c^{A}_J}{\lambda} \sum_{k=1}^N\hat{b}_{k-1,k}\right)\sum_{j=1}^N \left|y_{j+1}-y_j\right|\nonumber\\
\le&\left(C_AC_{|A|}c^{\tau}_J +\frac{C_J^A}{\lambda}\sum_{j=1}^N \left|y_{j+1}-y_j\right|\right)
C_{|A|}\sum_{j=1}^{N}\hat{b}_{j,j-1} \exp \left(\frac{C_A}{\lambda} \sum_{k=1}^N\hat{b}_{k-1,k}\right)\quad.\nonumber
\end{align}
\end{proof}

\paragraph{Acknowledgements}
We thank the DFG (Deutsche Forschungsgemeinschaft) and the HIM  (Hausdorff Research Institute for Mathematics) for supporting this project.



\bibliographystyle{spmpsci}      
\bibliography{LITERATUR}   

%
%

\end{document}